\documentclass[11pt]{article}
\usepackage{graphicx}
\usepackage{subcaption}
\usepackage{todonotes}
\usepackage{mathtools,booktabs}
\usepackage{soul}
\usepackage{amsthm}
\usepackage{amsmath,amssymb}
\usepackage{systeme}

\usepackage{gensymb}
\usepackage[ruled,vlined,linesnumbered]{algorithm2e}
\SetKwComment{Comment}{/* }{ */}
\SetKwInOut{Input}{input}\SetKwInOut{Output}{output}

\graphicspath{{./figures/}}
\usepackage[absolute]{textpos}

\usepackage{lineno}
\usepackage{fullpage,cite}
\usepackage[top=0.8in, bottom=0.9in, left=0.9in, right=0.9in]{geometry}

\usepackage[pdftex, plainpages = false, pdfpagelabels,
                 bookmarks=false,
                 bookmarksopen = true,
                 bookmarksnumbered = true,
                 breaklinks = true,
                 linktocpage,
                 pagebackref,
                 colorlinks = true,  % was true
                 linkcolor = blue,
                 urlcolor  = cyan,
                 citecolor = red,
                 anchorcolor = green,
                 hyperindex = true,
                 hyperfigures
                 ]{hyperref}

\def\bbR{\mathbb{R}}

\def\calB{\mathcal{B}}

\newtheorem{lemma}{Lemma}
\newtheorem{corollary}{Corollary}
\newtheorem{observation}{Observation}
\newtheorem{theorem}{Theorem}
\newtheorem{definition}{Definition}
\title{Shortest Path Map Equivalence Decompositions and Applications\thanks{A preliminary version of this paper will appear in {\em Proceedings of the 34th Annual European Symposium on Algorithms (ESA 2026)}.
}
}
\author{Haitao Wang\thanks{Kahlert School of Computing,
University of Utah, Salt Lake City, UT 84112, USA. {\tt haitao.wang@utah.edu}}
}
\date{}

\def\bbR{\mathbb{R}}
\def\calP{\mathcal{P}}
\def\calB{\mathcal{B}}
\def\calL{\mathcal{L}}

\def\calU{\mathcal{U}}

\def\calO{\mathcal{O}}
\def\calA{\mathcal{A}}

\def\Psispt{\Psi^{spt}}
\def\Psisptb{\Psi^{spt}_{\mathcal{B}}}
\def\spm{S\!P\!M}
\def\spt{S\!P\!T}
\def\vis{V\!i\!s}

\def\st{$s$-$t$}
\def\spmb{S\!P\!M_{\mathcal{B}}}
\def\Psibp{\Psi_{\calB}(\calP)}
\def\Psibb{\Psi_{\calB}(\calB)}
\def\Psipb{\Psi_{\calP}(\calB)}
\def\Psipp{\Psi_{\calP}(\calP)}

\def\visb{\vis_{\calB}}

\begin{document}

\maketitle

\vspace{-0.2in}
\begin{abstract}
Given a polygonal domain $\calP$ in the plane, the shortest path map with respect to a point $s$, denoted by $\spm(s)$, is the decomposition of $\calP$ into cells such that shortest paths from $s$ to all points $t$ in the same cell have the same vertex sequence. The shortest path map equivalence decomposition of $\calP$ is the decomposition of $\calP$ into cells so that $\spm(s)$ is topologically equivalent for all points $s$ in the same cell. In this paper, we prove new upper bounds on the combinatorial complexities of the $\spm$-equivalence decompositions under various settings, depending on whether $s$ and/or $t$ are restricted to the boundary of $\calP$. We also propose new algorithms to compute these decompositions. Further, our results lead to new solutions to several other problems, including answering two-point shortest path queries in $\calP$, and computing geodesic diameter and center of $\calP$. 
\end{abstract}


\section{Introduction}
\label{sec:intro}

Let $\calP$ be a polygonal domain of $h$ holes with a total of $n$ vertices in the plane, i.e., $\calP$ is a closed and multiply-connected region bounded by $n$ segments that form $h+1$ cycles (holes and the region outside $\calP$ are also called {\em obstacles}). The {\em shortest path map} with respect to a point $s\in \calP$, denoted by $\spm(s)$, is the decomposition of $\calP$ into cells such that shortest paths from $s$ to all points $t$ in the same cell are combinatorially the same (i.e., they have the same vertex sequence)~\cite{ref:ChiangTw99,ref:HershbergerAn99}; see Fig.~\ref{fig:spm}. The {\em shortest path map equivalence decomposition} (or {\em SPM-equivalence decomposition} for short) of $\calP$, denoted by $\Psi$, is the decomposition of $\calP$ into cells so that $\spm(s)$ is topologically equivalent for all points $s$ in the same cell. Chiang and Mitchell~\cite{ref:ChiangTw99} first studied the SPM-equivalence decomposition and used it to answer two-point shortest path queries in $\calP$. They proved that the combinatorial complexity of $\Psi$ is bounded by $O(n^{10})$ and proposed an $O(n^{10}\log n)$ time algorithm to compute it. 

We consider several variants of the SPM-equivalence decompositions, depending on whether $s$ and/or $t$ are required to be on the boundary of $P$. Specifically, let $\calB$ denote the boundary of $\calP$. 
We define the {\em boundary-restricted} shortest path map for a point $s$, denoted by $\spmb(s)$, as the decomposition of $\calB$ into segments such that shortest paths from $s$ to all points $t$ in the same segment are combinatorially the same. 
We define the {\em boundary-restricted SPM-equivalence decomposition} with respect to $\calP$, denoted by $\Psibp$, as the decomposition of $\calB$ into segments so that $\spm(s)$ are topologically equivalent for all points $s$ in the same segment. Similarly, we let $\Psibb$ denote the decomposition of $\calB$ into segments so that $\spmb(s)$ is topologically the same for all points $s$ in the same segment, and $\Psipb$ the decomposition of $\calP$ into cells so that $\spmb(s)$ is topologically the same for all points $s$ in the same cell. If we follow the same convention for notation, then $\Psipp$ is $\Psi$ (we will use the two notations interchangeably). Another way to view these four decompositions is to define them depending on whether $s$ and/or $t$ are restricted to be on $\calB$. For example, $\Psipb$ is for the case where $s\in \calP$ while $t\in \calB$. Note that $\Psipp$ is a refinement of $\Psipb$ (i.e., each cell of $\Psipp$ is completely inside a cell of $\Psipb$), and $\Psibp$ is a refinement of $\Psibb$.  

\begin{figure}[t]
\begin{minipage}[t]{\textwidth}
\begin{center}
\includegraphics[height=1.8in]{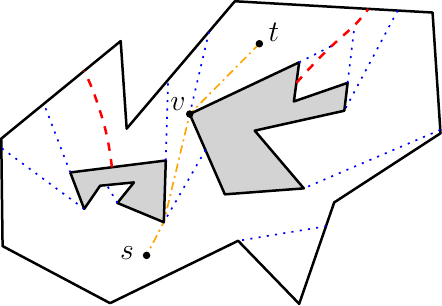}
\caption{ Illustrating a shortest path map $\spm(s)$: The dotted blue segments are extension segments and the dashed red curves are bisector curves. The point $t$ is in a cell whose root is $v$. }
\label{fig:spm}
\end{center}
\end{minipage}
\vspace{-0.1in}
\end{figure}

While we are not aware of any previous work about $\Psibb$ or $\Psibp$, 
%it might be possible to obtain an upper bound $O(n^5\log n\log^*n)$ for $\Psibb$ using the techniques of Bae and Okamoto~\cite{ref:BaeQu12} and also compute $\Psibb$ in $O(n^5\log n\log^*n)$ time. The decomposition 
$\Psipb$ was implicitly discussed in \cite{ref:ChiangTw99}. Chiang and Mitchell's approach~\cite{ref:ChiangTw99} also gave an $O(n^{10})$ upper bound for $|\Psipb|$, asymptotically the same as $|\Psipp|$. Their algorithm computes $\Psipb$ in $O(n^{10}\log n)$ time. 
%i.e., their method showed that $|\Psipb|=O(n^{10})$ and $\Psipb$ can be computed in $O(n^{10}\log n)$ time. 

In this paper, we obtain new bounds and algorithms for these decompositions (see Table~\ref{tab:decomresult} for a summary). 

\begin{itemize}
\item For $\Psibb$, we prove $|\Psibb|=O(n^{4+\epsilon})$ and give an algorithm to compute $\Psibb$ in $O(n^{4+\epsilon})$ time; throughout the paper, let $\epsilon$ represent an arbitrarily small positive constant. For comparison, a traditional method~\cite{ref:BaeTh13,ref:MitchellSh96,ref:WangOn18} can give an $O(n^5)$ upper bound for $|\Psibb|$ and compute it in $O(n^5\log n)$ time.
\item 
For $\Psibp$, we present an algorithm to compute it in $O((n^2+|\Psibp|)\cdot \log n)$ time, and prove 
 that $|\Psibp|=O(n^{5})$. Note that $|\Psibp|$ denotes the combinatorial size of $\Psibp$ in the {\em worst case} among all input instances $\calP$ (not the size in any specific input instance). We follow the same convention for other similar notation, e.g., $|\Psipp|$, $|\Psibb|$, $|\Psipb|$. 
 Therefore, a time complexity like $O((n^2+|\Psibp|)\cdot \log n)$ is a function of the worst-case size of $\Psibp$.
 %and we say that it is {\em decomposition-sensitive}. 
 
\item 
For the most general decomposition $\Psipp$, we do not have a better bound than $O(n^{10})$  proved by Chiang and Mitchell~\cite{ref:ChiangTw99}. However, we present an algorithm that can compute $\Psipp$ in 
%a decomposition-sensitive 
$O(n^{7.73}+|\Psipp|\log n)$ time. 
%(or $O(n^7(h+\log n)+|\Psipp|\cdot \log n)$ time for small $h$). 
As such, if $|\Psipp|$ could be bounded by $o(n^{10})$, then our algorithm for constructing $\Psipp$ would be faster than the $O(n^{10}\log n)$-time solution in \cite{ref:ChiangTw99}. 
Note that the time complexity of the algorithm of \cite{ref:ChiangTw99} does not depend on $|\Psipp|$.
%Further, if $|\Psipp|=\Omega(n^7)$, then our algorithm would be optimal within a logarithmic factor. 
The currently best-known lower bound for $|\Psipp|$ is $\Omega(n^4)$~\cite{ref:ChiangTw99}.
Hence, our result reduces the problem of designing efficient algorithms for computing $\Psipp$ to proving smaller upper bounds for its combinatorial complexity.

\item 
For $\Psipb$, we prove $|\Psipb|=O(n^{8})$ and present an algorithm that can compute it in $O(n^5\log^2 n+|\Psipb|\cdot \log n)$ time. 
%As above, our contribution is to reduce the problem of designing efficient algorithms for computing $\Psipb$ to proving smaller upper bounds for its combinatorial complexity. 
\end{itemize}

\begin{table}[h]
    \caption{Summary of the results on the four decompositions.}
    \label{tab:decomresult}
    \centering
    \tabcolsep1ex%
    %\resizebox{\linewidth}{!}%
    \vspace{0.15in}
    {%
        \begin{tabular}{l|ll|ll}
         \toprule
              & \multicolumn{2}{c|}{Combinatorial Sizes} & \multicolumn{2}{c}{Algorithm Complexities}       \\            
              &Previous & Ours & Previous & Ours       \\
            %\midrule
            %\specialrule{1.5pt}
            \hline
            $\Psibb$ & $O(n^5)$~\cite{ref:BaeTh13,ref:MitchellSh96,ref:WangOn18} &  $O(n^{4+\epsilon})$ &  $O(n^5)$~\cite{ref:BaeTh13,ref:MitchellSh96,ref:WangOn18} &  $O(n^{4+\epsilon})$ \\
            $\Psibp$ &  &  $O(n^{5})$ &  &  $O((n^2+|\Psibp|)\cdot \log n)$ \\
            $\Psipp$ & $O(n^{10})$~\cite{ref:ChiangTw99} &   &  $O(n^{10}\log n)$~\cite{ref:ChiangTw99} &  $O(n^{7.73}+|\Psipp|\cdot \log n)$ \\
            $\Psipb$ & $O(n^{10})$~\cite{ref:ChiangTw99} &  $O(n^{8})$ &  $O(n^{10}\log n)$~\cite{ref:ChiangTw99} &  $O(n^5\log^2 n+|\Psipb|\cdot \log n)$ \\
            \bottomrule
        \end{tabular}}
\end{table}

%It is remarkable that our algorithms for all these decompositions are size-sensitive. To the best of our knowledge, we are not aware of any size-sensitive algorithms for constructing these decompositions in the previous work. 

With our results on these decompositions, we further obtain new results on answering two-point shortest path queries, and computing geodesic diameter and center, as discussed below. 
%We believe they will find many other applications as well.

% \begin{table}[h]
%     \caption{Summary of the results on the four decompositions.}
%     \label{tab:decomresult}
%     \centering
%     \tabcolsep1ex%
%     %\resizebox{\linewidth}{!}%
%     \vspace{0.15in}
%     {%
%         \begin{tabular}{lc|ccc}
%          \toprule
%               %& \multicolumn{2}{c|}{Combinatorial Sizes} & \multicolumn{2}{c}{Algorithm Complexities}       \\            
%               &  & $\calB$-$\calB$ & $\calB$-$\calP$ & $\calP$-$\calP$       \\
%             %\midrule
%             %\specialrule{1.5pt}
%             \hline
%             Shortest path  & previous & $O(n^{4+\epsilon},n^{4+\epsilon},\log n)$~\cite{ref:deBergTo24} & $O(n^{6+\epsilon},n^{6+\epsilon},\log^2 n)$~\cite{ref:deBergTo24}  & $O(n^{10+\epsilon},n^{10+\epsilon},\log n)$~\cite{ref:deBergTo24} \\
%             queries & ours & $O(|\Psibb|\log n,|\Psibb|,\log n)$  & $O(n^5\log n,n^5\log n,\log^2 n)$  &  \\
%                    &   &         & $O(n^6,n^6,\log n)$  &  \\
%             \hline
%             \multirow{ 2}{*}{Diameter}  & previous &  &   &  \\
%              & ours &  &   &  \\
%              \hline
%             \multirow{ 2}{*}{Center}  & previous &  &   &  \\
%              & ours &  &   &  \\
%             \bottomrule
%         \end{tabular}}
% \end{table}

\paragraph{Two-point shortest path queries.}
The problem is to build a data structure so that a (Euclidean) shortest path from $s$ to $t$ in $\calP$ can be quickly found for any two query points $s$ and $t$. Using their SPM-equivalence decomposition, Chiang and Mitchell~\cite{ref:ChiangTw99} constructed a data structure in $O(n^{11})$ space and preprocessing time, and each query can be answered in $O(\log n)$ time (this is the time for computing the length of the shortest path; reporting the path itself needs additional time linear in the number of edges of the path; all two-point shortest path query data structures discussed in the paper have this feature and therefore we will not repeat this again). For simplicity, we use $O(P(n),S(n),Q(n))$ to denote the complexity of such a data structure, where $P(n)$, $S(n)$, and $Q(n)$ denote the preprocessing time, the space, and the query time of the data structure, respectively. Chiang and Mitchell~\cite{ref:ChiangTw99} also gave a data structure of $O(n^{10}\log n, n^{10}\log n, \log^2n)$ complexity. In case the number of obstacles $h$ is much smaller than the number of vertices $n$, they also provided data structures with complexities sensitive to $h$. For example, they gave a data structure of complexity $O(n^5,n^5,h+\log n)$ and another data structure of $O(n+h^5)$ space with $O(h\log n)$ query time. Refer to \cite{ref:ChiangTw99} for other tradeoffs between preprocessing and query time. 

Guo, Maheshwari, and Sack~\cite{ref:GuoSh08} also studied the problem and presented a data structure of complexity $O(n^2\log n,n^2,h\log n)$. Bae and Okamoto~\cite{ref:BaeQu12} considered a special case where both query points are restricted on $\calB$, and they derived a data structure of complexity $O(n^5\log n\log^*n,n^5\log^*n,\log n)$. By introducing a novel idea of using cuttings~\cite{ref:ChazelleCu93}, de Berg, Miltzow, and Staals~\cite{ref:deBergTo24} revisited the problem and gave two data structures of complexities $O(n^{10+\epsilon},n^{10+\epsilon},\log n)$ and $O(n^{9+\epsilon},n^{9+\epsilon},\log^2 n)$, respectively; their preprocessing algorithm involves randomization (due to using the algorithm in~\cite{ref:KoltunAl04}) and the preprocessing time is expected, while the space and query complexities are worst-case. Note that although cuttings are not new, using cuttings to solve two-point shortest path queries is novel. For the case where both $s$ and $t$ are restricted on $\calB$ (referred to as the {\em $\calB$-$\calB$ case}), they gave a data structure of $O(n^{4+\epsilon},n^{4+\epsilon},\log n)$ complexity~\cite{ref:deBergTo24}. If only one of the query points is restricted on $\calB$ (referred to as the {\em $\calB$-$\calP$ case}), they gave a randomized data structure of $O(n^{6+\epsilon},n^{6+\epsilon},\log^2 n)$ complexity,
%\footnote{The authors claim that their query time is $O(\log n)$. However, they used the data structure from~\cite{ref:AgarwalCo97}, whose query time is $O(\log^2 n)$.}, 
where the preprocessing time is expected. It should be noted that the method of \cite{ref:deBergTo24}, which does not rely on SPM-equivalence decompositions, does not yield any new results on the SPM-equivalence decompositions. 

Building upon our new results on the SPM-equivalence decompositions, we achieve the following results on two-point shortest path queries (see Table~\ref{tab:sp} for a summary). 

\begin{itemize}
\item 
For the $\calB$-$\calB$ case, we obtain a data structure of complexity $O(|\Psibb|\log n,|\Psibb|,\log n)$, assuming that $\Psibb$ is available. With our $O(n^{4+\epsilon})$ time algorithm for computing $\Psibb$, we obtain the complexity $O(n^{4+\epsilon},n^{4+\epsilon},\log n)$ for the data structure. This matches the result of \cite{ref:deBergTo24}.
%This improves the $O(n^{4+\epsilon},n^{4+\epsilon},\log n)$ result from~\cite{ref:deBergTo24} when $h=o(n)$.

\item 
For the $\calB$-$\calP$ case, we obtain two data structures of complexities $O(n^2\log n + n\cdot |\Psibp|,n\cdot |\Psibp|,\log n)$ and $O((n^2+|\Psibp|)\cdot \log n,|\Psibp|\cdot \log n,\log^2 n)$, respectively. Using our bound $|\Psibp|=O(n^5)$, these are $O(n^6,n^6,\log n)$ and $O(n^5\log n,n^5\log n,\log^2 n)$, respectively. This improves the result of $O(n^{6+\epsilon},n^{6+\epsilon},\log^2 n)$ from~\cite{ref:deBergTo24}. 

\item 
For the general case (i.e., the $\calP$-$\calP$ case), 
%following Chiang and Mitchell's approach~\cite{ref:ChiangTw99} and using our new algorithm for constructing the decomposition $\Psipp$, 
we can have two data structures of complexities $O(n^{7.73}+|\Psipp|\log n,(n^7+|\Psipp|)\log n,\log^2 n)$ and $O(n^{8}+n\cdot |\Psipp|,n^8+n\cdot |\Psipp|,\log n)$, respectively.   
%$O(n^{7.73}+|\Psipp|\log n)$ time (or $O(n^7(h+\log n)+|\Psipp|\cdot \log n)$ time for small $h$)
Using the current best bound $|\Psipp|=O(n^{10})$, these match the results from \cite{ref:ChiangTw99} and are worse than the results from \cite{ref:deBergTo24}. However, if the ``true'' size of $|\Psipp|$ is much smaller than $O(n^{10})$ (e.g., $|\Psipp|=O(n^9)$), then our result could be better than all previous work. 
As such, our result reduces the problem of designing a more efficient algorithm to proving a smaller bound for $|\Psipp|$. 
\end{itemize}
%In addition, for a special case where each of $s$ and $t$ is restricted to the boundary of one particular obstacle, we give a data structure of complexity $O(n^{3+\epsilon},n^{3+\epsilon},\log n)$. For another special case where $s$ is restricted to a given algebraic curve $\gamma$ of constant degree and possibly intersecting obstacles of $\calP$, and $t$ is restricted to $\calB$ (or another given algebraic curve $\gamma$ of constant degree and possibly intersecting obstacles of $\calP$), we derive a data structure of complexity $O(n^{3.5+\epsilon},n^{3.5+\epsilon},\log n)$.

\begin{table}[h]
    \caption{Summary of the results on two-point shortest path queries. 
    %The currently best-known bounds for $|\Psibb|$, $|\Psibp|$, and $|\Psipp|$ are $O(n^{4+\epsilon})$, $O(n^5)$, and $O(n^{10})$, respectively.
    }
    \label{tab:sp}
    \centering
    \tabcolsep1ex%
    %\resizebox{\linewidth}{!}%
    \vspace{0.15in}
    {%
        \begin{tabular}{l|l|l}
         \toprule              
              &Previous & Ours       \\
            %\midrule
            %\specialrule{1.5pt}
            \hline
            $\calB$-$\calB$ case & $O(n^{4+\epsilon},n^{4+\epsilon},\log n)$~\cite{ref:deBergTo24} & $O(n^{4+\epsilon},n^{4+\epsilon},\log n)$  \\            
            \hline
            $\calB$-$\calP$ case & $O(n^{6+\epsilon},n^{6+\epsilon},\log^2 n)$~\cite{ref:deBergTo24} & $O(n^5\log n,n^5\log n,\log^2 n)$ or  $O(n^6,n^6,\log n)$ \\   
            \hline
            $\calP$-$\calP$ case & $O(n^{10+\epsilon},n^{10+\epsilon},\log n)$~\cite{ref:deBergTo24} & $O(n^{8}+n\cdot |\Psipp|,n^8+n\cdot |\Psipp|,\log n)$  \\     
                                 & $O(n^{9+\epsilon},n^{9+\epsilon},\log^2 n)$~\cite{ref:deBergTo24} & $O(n^{7.73}+|\Psipp|\log n,(n^7+|\Psipp|)\log n,\log^2 n)$  \\                        
            \bottomrule
        \end{tabular}}
\end{table}

Computing shortest paths in polygonal domains is a classical problem and has been studied extensively in the past several decades, e.g.,~\cite{ref:ChenCo13,ref:Eriksson-BiqueGe15,ref:GhoshAn91,ref:HershbergerAn99,ref:HershbergerA22,ref:MitchellA91,ref:MitchellSh96,ref:StorerSh94,ref:RohnertSh86,ref:GuibasOp89,ref:GuibasLi87,ref:HershbergerA91,ref:HershbergerCo94,ref:LeeEu84}.
For the one-point query problem where $s$ is given in the input and $t$ is the only query point, the problem becomes much easier and $\spm(s)$, which is of space $O(n)$, along with a point location data structure~\cite{ref:EdelsbrunnerOp86,ref:KirkpatrickOp83,ref:SarnakPl86}, can be used to answer shortest path queries in $O(\log n)$ time. Constructing $\spm(s)$ can be done in $O(n\log n)$ time~\cite{ref:HershbergerAn99} or in $O(n+h\log h)$ time~\cite{ref:WangA23} after $\calP$ is triangulated. If $\calP$ is a simple polygon (i.e., $h=0$), then there exists a data structure of $O(n,n,\log n)$ complexity for the two-point query problem~\cite{ref:GuibasOp89}.

\paragraph{Geodesic diameter.}
For any pair of points in $\calP$, the length of their shortest path in $\calP$ is also called the {\em geodesic distance}.  
The pair of points of $\calP$ whose geodesic distance is the largest is called a {\em diametral pair} and their geodesic distance is called the {\em geodesic diameter} of $\calP$. %The {\em geodesic diameter problem} is to compute the diameter of $\calP$ (as well as find a diametral pair). 

If $\calP$ is a simple polygon, computing the geodesic diameter can be done in $O(n)$ time~\cite{ref:HershbergerMa97}, improving the $O(n\log n)$ time results~\cite{ref:GuibasOp89,ref:SuriCo89}. For the polygonal domain case, the problem becomes much more challenging. One main difficulty lies in that a diametral pair of points could both be in the interior of $\calP$. Bae, Korman, and Okamoto~\cite{ref:BaeTh13} thoroughly studied the problem and discovered many interesting properties. Their effort led to an algorithm of $O(n^{7.73})$ time (or $O(n^7(h+\log n))$ for small $h$) for computing the diameter. Their algorithm can be improved to $O(n^{7.4+\epsilon})$ time using the new shortest path query algorithm~\cite{ref:deBergTo24}. 
Other restricted cases of the problem were also studied in \cite{ref:BaeTh13}. For the $\calB$-$\calB$ case, which aims to find a {\em restricted} diametral pair $(s,t)$ with $s\in \calB$ and $t\in \calB$, the method of \cite{ref:BaeTh13} solved the problem in $O(n^5\log n\log^*n)$ time. %(or $O(n^5(h+\log n))$ for small $h$). 
For the $\calB$-$\calP$ case, where one wishes to find a restricted diametral pair $(s,t)$ with $s\in \calB$ and $t\in \calP$, an algorithm of roughly $O(n^{5.91})$ time 
%(or $O(n^5(h+\log n))$ for small $h$) 
is presented in \cite{ref:BaeTh13}. 
%In part due to the difficulty of the problem, no progress has been made since the work~\cite{ref:BaeTh13} was published more than a decade ago. 

Based on our new results on the SPM-equivalence decompositions, we solve the $\calB$-$\calB$ case in $O(n^{4+\epsilon})$ time and the $\calB$-$\calP$ case in $O(n^{5}\log n)$ time.
These improve the $O(n^5\log n\log^*n)$ time and $O(n^{5.91})$ time results in \cite{ref:BaeTh13}, respectively (see Table~\ref{tab:diameter} for a summary; note that the $\calP$-$\calB$ case is symmetric to the $\calB$-$\calP$ case).

\begin{table}[h]
    \caption{Summary of the results on the geodesic diameter and geodesic center.}
    \label{tab:diameter}
    \centering
    \tabcolsep1ex%
    %\resizebox{\linewidth}{!}%
    \vspace{0.15in}
    {%
        \begin{tabular}{l|ll|ll}
         \toprule
              & \multicolumn{2}{c|}{Geodesic Diameter} & \multicolumn{2}{c}{Geodesic Center}       \\            
              &Previous & Ours & Previous & Ours       \\
            %\midrule
            %\specialrule{1.5pt}
            \hline
            $\calB$-$\calB$ case & $O(n^5\log n\log^*n)$~\cite{ref:BaeTh13} &  $O(n^{4+\epsilon})$ &  $O(n^8\log n)$~\cite{ref:WangOn18} &  $O(n^{4+\epsilon})$ \\
            $\calB$-$\calP$ case & $O(n^{5.91})$~\cite{ref:BaeTh13}  &  $O(n^{5}\log n)$ & $O(n^8\log n)$~\cite{ref:WangOn18} &  $O(n^5\log n\log^* n)$ \\
            $\calP$-$\calB$ case & $O(n^{5.91})$~\cite{ref:BaeTh13}  &  $O(n^{5}\log n)$ & $O(n^{11}\log n)$~\cite{ref:WangOn18} &  $O(n^{10+\epsilon})$ \\
            $\calP$-$\calP$ case & $O(n^{7.4+\epsilon})$~\cite{ref:deBergTo24} &   &  $O(n^{11}\log n)$~\cite{ref:WangOn18} &  $O((n^7+|\Psipp|)\cdot n^{2+\epsilon})$ \\
            \bottomrule
        \end{tabular}}
\end{table}

%In addition, if each of $s$ and $t$ is restricted to the boundary of a single obstacle, then our algorithm can find such a diametral pair in $O(n^{3+\epsilon})$ time. Using our new bound and algorithm for the decomposition $\Psibp$, we show that this problem can be solved in $O(n^{5}\log n)$ time. Our approach is quite flexible and can solve other variations of the problem. For example, if $s$ is restricted on $\gamma$ and $t$ is restricted on $\gamma'$ for two curves $\gamma$ and $\gamma'$ as defined above, then such a diametral pair $(s,t)$ can be found in $O(n^{3.5+\epsilon})$ time. 

\paragraph{Geodesic center.}
A closely related concept is
the {\em geodesic center} of $\calP$, which is a point that minimizes the maximum geodesic distance from it to any other point in $\calP$. The problem of computing the geodesic center has attracted much attention~\cite{ref:AhnA16,ref:AsanoCo85,ref:BaeCo15CGTA,ref:PollackCo89,ref:WangOn18}. If $\calP$ is a simple polygon, a linear-time algorithm for computing a geodesic center is known~\cite{ref:AhnA16,ref:LubiwTh25}. For the polygonal domain case, the problem again becomes significantly more challenging. One reason is that a farthest point of a point may be in the interior of $\calP$~\cite{ref:BaeTh13}; this is in contrast to the simple polygon case in which a farthest point must be a vertex of $\calP$. Bae, Korman, and Okamoto~\cite{ref:BaeCo15CGTA} gave the first-known algorithm that can compute a geodesic center in $O(n^{12+\epsilon})$ time. Later Wang~\cite{ref:WangOn18} discovered a so-called {\em $\pi$-range property}, which eventually led to an $O(n^{11}\log n)$ time algorithm. 

We obtain improved results on several restricted versions of the geodesic center problem (see Table~\ref{tab:diameter} for a summary). 
\begin{itemize}
    \item First, if the center is required to be on $\calB$ and farthest points are also required to be on $\calB$, i.e., finding a point on $\calB$ to minimize the maximum geodesic distance from it to any other point in $\calB$, then we solve this case in 
    $O(|\Psibb|\log n\log^* n)$ time, which is $O(n^{4+\epsilon})$ with our bound $|\Psibb|=O(n^{4+\epsilon})$.     
    For comparison, the method of \cite{ref:WangOn18} solves this problem in $O(n^8\log n)$ time. 
    \item Second, if the center is required on $\calB$ but farthest points can be anywhere in $\calP$, then our algorithm runs in $O(|\Psibp|\log n\log^* n)$ time, which is 
    $O(n^5\log n\log^*n)$ with our bound $|\Psibp|=O(n^5)$.  
    The algorithm of \cite{ref:WangOn18} solves this problem in $O(n^8\log n)$ time.     
    \item Third, if the center can be anywhere in $\calP$ while the farthest points are required to be on $\calB$, 
    %i.e., finding a point in $\calP$ to minimize the maximum geodesic distance from it to any point in $\calB$, 
    then we solve the problem in $O((n^5+|\Psipb|)\cdot n^{2+\epsilon})$ time, which is $O(n^{10+\epsilon})$ with our bound $|\Psipb|=O(n^{8})$. The method of \cite{ref:WangOn18} solves this problem in $O(n^{11}\log n)$ time.
    %As discussed above, the currently best known upper bound for $|\Psipb|$ is $O(n^{10})$~\cite{ref:ChiangTw99}.
    \item Finally, for the most general problem, our algorithm runs in $O((n^7+|\Psipp|)\cdot n^{2+\epsilon})$ time. As discussed above, the currently best known upper bound for $|\Psipp|$ is $O(n^{10})$~\cite{ref:ChiangTw99}.
    %which would provide an improvement over the $O(n^{11}\log n)$ time result in \cite{ref:WangOn18} if $|\Psipp|$ could be bounded by $O(n^9)$ time. As discussed above, the currently best known upper bound for $|\Psipp|$ is $O(n^{10})$~\cite{ref:ChiangTw99}.
\end{itemize}

\subsection{An overview of our approach}
To prove an upper bound for $|\Psibb|$, we reduce the problem to analyzing the combinatorial complexities of lower envelopes of certain functions. For each vertex $u$ of $\calP$, we define a set of functions $f_u(s,t)$ as follows. 
Let $\vis(u)$ denote the visibility polygon of $u$ in $\calP$. For each edge $e_t\in \spmb(u)$, $e_t$ lies in a single cell of $\spm(u)$, and let $v_t$ be the root of the cell. For each edge $e_s$ of $\vis(u)\cap \calB$, we define a function $f_u(s,t)=|su|+d(u,v_t)+|v_tt|$, for $s\in e_s$ and $t\in e_t$, where $d(u,v_t)$ is the geodesic distance between $u$ and $v_t$. Since both $\vis(u)\cap \calB$ and $\spmb(u)$ has $O(n)$ edges, the above defines $O(n^2)$ bivariate functions (of constant degree) for $u$, each of which defines a constant-sized two-dimensional surface patch in $\bbR^3$. As $\calP$ has $O(n)$ vertices, we have $O(n^3)$ functions in total; let $F$ denote the set of all of these functions. Our first observation is that a vertex of $\Psibb$ corresponds to a vertex of the lower envelope $\calL(F)$ of (the surface patches defined by) the functions of $F$. As such, it suffices to analyze the complexity of $\calL(F)$.

%determines the complexity of $\Psibb$, and more specifically, $|\Psibb|=\Theta(\calL(F))$. 

To find a bound for $|\calL(F)|$, one could apply the result of Sharir~\cite{ref:SharirAl94}, who proved an $O(n^{d+\epsilon})$ upper bound for the size of the lower envelope of a collection of $n$ constant-sized (partially defined) $d$-variate algebraic functions of constant degree, for any $d\geq 2$ (the $d=2$ case was first proved in~\cite{ref:HalperinNe94}; see also \cite[Chapter 7]{ref:SharirDa95}). 
Since $|F|=O(n^3)$ and $d=2$ in our problem, applying the above result directly could obtain an upper bound $O(n^{6+\epsilon})$ for $|\calL(F)|$. We instead take a closer look at the problem. Our main idea is that we find a way to decompose the lower envelope $\calL(F)$ into $O(n^2)$ regions so that for each region it suffices to consider $O(n)$ functions of $F$. Consequently, by applying the algorithm of~\cite{ref:HalperinNe94,ref:SharirAl94,ref:SharirDa95} for each region, we can prove an upper bound of $O(n^{4+\epsilon})$ for $|\calL(F)|$ and thus for $|\Psibb|$, and also compute the vertices of $\calL(F)$ and thus compute $\Psibb$ in $O(n^{4+\epsilon})$ time. Note that traditional brute-force method~\cite{ref:BaeTh13,ref:MitchellSh96,ref:WangOn18} can give an $O(n^5)$ upper bound for $|\Psibb|$ and compute it in $O(n^5\log n)$ time.

For $\Psibp$, it turns out that a simple brute-force approach can provide a good upper bound for its size (similar techniques were already used in previous work for related problems~\cite{ref:BaeTh13,ref:ChiangTw99,ref:WangOn18}). 
%Computing the other two decompositions $\Psibp$ and $\Psibb$, however, are done using ``direct'' methods. 
To compute $\Psibp$, since it is a decomposition of $\calB$, which is one-dimensional, we use a sweeping point algorithm, i.e., moving a point $s$ on $\calB$ and find those event points when $\spm(s)$ changes combinatorially. 

For $\Psipp$, unfortunately, our approach does not yield a better result than the $O(n^{10})$ upper bound from \cite{ref:ChiangTw99}. One reason is that the decomposition is on $\calP$, which is two-dimensional, and a vertex of it does not necessarily correspond to a vertex of the lower envelope of the corresponding functions (see Section~\ref{sec:psipp} for a detailed explanation). Hence, analyzing the complexity of the lower envelope does not immediately lead to an upper bound for $\Psipp$. The same issue happens to $\Psipb$ as well. 
In contrast, both $\Psibb$ and $\Psibp$ are decompositions of $\calB$, which is essentially one-dimensional, and a vertex of them corresponds to a vertex of the corresponding lower envelope. 

%Although we do not have improved upper bounds for the sizes of $\Psipp$ and $\Psipb$, we propose size-sensitive algorithms for computing them. To this end, 

To construct $\Psipp$ and $\Psipb$, we first present efficient algorithms to compute the edges of the lower envelopes of the corresponding functions. We then compute these decompositions by using the edges of the lower envelopes.
%in a decomposition-sensitive manner. 
%However, by adapting the Agarwal, Aronov, and Sharir's technique for computing the lower bound for algebraic functions~\cite{ref:AgarwalCo97}, we are able to derive an output-sensitive algorithm for computing $\Psipp$. Again, our algorithm is not a direct application of the algorithm. Indeed, adapting the technique requires many new observations and insight into our problem. Also, the algorithm in~\cite{ref:AgarwalCo97} is randomized while ours is deterministic. Computing $\Psipb$ can be done in a similar manner. 

%The output-sensitive algorithms for constructing $\Psibb$ and $\Psibp$ make uses of the property that the object we need to decompose is $\calB$, which is essentially a one-dimensional space. We use a sweeping point algorithm to construct each of the decomposition, i.e., moving a point $s$ on $\calB$ and find those event points when $\spmb(s)$ (resp., $\spm(s)$) is experiencing changes combinatorially. 

%Our new bounds and algorithms for SPM-equivalent decompositions lead to new results on the two-point shortest path query problem, geodesic diameter and center problems in polygonal domains.

For answering two-point shortest path queries, using a SPM-equivalent decomposition, a data structure can be built by constructing a ``parameterized'' shortest path map for points in each cell of the decomposition, which is the approach first proposed in~\cite{ref:ChiangTw99} (a linear factor in the preprocessing can be saved by using persistent data structures~\cite{ref:DriscollMa89}, but the query time grows to $O(\log^2 n)$).  
This approach works for all SPM-equivalent decompostions, e.g., for the $\calB$-$\calB$ case,  $\Psibb$ is used (in this case, we actually show that $O(\log n)$ query time can still be achieved even when a persistent data structure is utilized). 

For computing the geodesic diameter, using the results of Bae, Korman, and Okamoto~\cite{ref:BaeTh13}, we show that a diametral pair corresponds to a vertex of the lower envelope of the corresponding functions. 
%(or corresponds to a vertex of the decompositions).
%For example, if we are looking for a restricted diametral pair $(s,t)$ such that $s\in \calB$ and $t\in \calP$, then such a diametral pair corresponds to a vertex of in the lower envelope of the functions for $\Psipb$. 
Consequently, once we construct the lower envelope, the diametral pairs can be easily found. This approach works for all cases. 

%In particular, a diametral pair $(s,t)$ of the most general problem, i.e., $s\in \calP$ and $t\in \calP$, corresponds to a vertex in the lower envelope $\calL(F)$ of the collection $F$ of the functions for $\Psipp$. The number of vertices of $\calL(F)$ is $O(n^7)$ and our algorithm can compute all vertices in $O(n^7\log n)$ time and compute all diametral pairs in $O(n^7\log n)$ time. Note that the approach of Bae, Korman, and Okamoto~\cite{ref:BaeTh13} also proves that the number of all diametral pairs is $O(n^7)$. Their algorithm computes a superset of all diametral pairs in $O(n^7\log n)$ time. However, it takes additional $O(n^{7.73})$ time to identify the true diametral pairs from the superset. Our algorithm, in contrast, can compute all true diametral pairs directly in a total of $O(n^7\log n)$ time. 

For computing the geodesic centers, an observation is that for any point $s\in \calP$, its farthest point in $\calP$ must be a vertex of $\spm(s)$~\cite{ref:BaeCo15CGTA}. Based on this property, for each cell $\sigma$ of the corresponding SPM-equivalent decomposition, for each vertex $v$ of $\spm(s)$, we parameterize the geodesic distance $d(s,v)$ using the coordinate of $s\in \sigma$. Then, a geodesic center $s$ restricted to $s\in \sigma$ corresponds to a lowest vertex in the upper envelope of the functions $d(s,v)$ for all vertices $v$ of $\spm(s)$. This approach has been used in the previous work~\cite{ref:BaeCo15CGTA} for the general case. 
%which only considered the most general geodesic center by using the decomposition $\Psipp$. 
However, to make the algorithm efficient, one needs good upper bounds on SPM-equivalent decompositions and efficient algorithms to compute them. As we have new upper bounds and algorithms for the decompostions, we obtain new algorithms for computing the geodesic centers. 
%Note that which decomposition to use depends on the problem setting.  For example, if one is interested in a geodesic center $s$ restricted to $\calB$ with respect to all points in $\calP$, then the decomposition $\Psibp$ should be used. 

\paragraph{Outline.} The rest of the paper is organized as follows. Section~\ref{sec:pre} introduce some notation and concepts that will be used throughout the paper. Our results for SPM-equivalent decompositions $\Psibb$, $\Psibp$, $\Psipp$, and $\Psipb$ are 
are presented in the subsequent Sections~\ref{sec:psibb}, \ref{sec:psibp}, \ref{sec:psipp}, and \ref{sec:psipb}, respectively. 
The two-point shortest path query problem, the geodesic diameter problem, and the geodesic center problem are discussed in the Sections~\ref{sec:query}, \ref{sec:diameter}, and \ref{sec:center}, respectively. 

\section{Preliminaries}
\label{sec:pre}
We define some notation and concepts that will be used throughout the paper, in addition to those already introduced in Section~\ref{sec:intro}. 

For any two points $p$ and $q$ in the plane, denote by $\overline{pq}$ the line segment with $p$ and $q$ as endpoints; denote by $|pq|$ the length of the segment. We say that $p$ is {\em visible} to $q$ if $\overline{pq}\subseteq \calP$.

For any two points $s$ and $t$ in $\calP$, we use $\pi(s,t)$ to denote a shortest path from $s$ to $t$ in $\calP$. 
%In the casewhere shortest paths are not unique, $\pi(s,t)$ may refer to an arbitrary one. 
Denote by $d(s,t)$ the (Euclidean) length of $\pi(s,t)$;
we call $d(s,t)$ the {\em geodesic distance} between $s$ and $t$. The vertex of $\pi(s,t)$ adjacent to $s$ is called the {\em anchor} of $s$; the anchor of $t$ is defined similarly. 
%We use $\pi_{u,v}(s,t)$ to represent a shortest \st\ path that is the union of $\overline{su}$, a shortest path from $u$ to $v$, and $\overline{vt}$, for two vertices $u,v$ of $\calP$.
%For two line segments $e$ and $f$ in $\calF$, their {\em geodesic distance} is defined to be the minimum geodesic distance between any point on $e$ and any point on $f$, i.e., $\min_{s\in e, t\in f}d(s,t)$; by slightly abusing the notation, we use $d(e,f)$ to denote their geodesic distance.
%For any path $\pi$ in $\calP$, we use $|\pi|$ to denote its length.
%Consider two obstacle vertices $u$ and $v$. For any point $s\in \vis(u)$ and any point $t\in \vis(v)$, 

Define $d_{u,v}(s,t) = |{su}| + d(u,v) + |{vt}|$ for two vertices $u,v$ of $\calP$ and two points $s,t\in \bbR^2$; note that when using this notation, we do not always guarantee that $s$ is visible to $u$ and $t$ is visible to $v$, and $s$ or/and $t$ may not even be inside $\calP$. 

%such that $s$ is visible to $u$ and $v$ is visible to $t$. 

For any compact region $R$ in the plane, let $\partial R$ denote its boundary. 
%We use $\partial \calP$ to denote the union of the boundaries of all obstacles of $\calP$. 
Recall that $\calB$ is the boundary of $\calP$. 
%Define $\calI=\calP\setminus\calB$, i.e., $\calI$ consists of all interior points of $\calP$. 
%Define $\calV$ as the set of all vertices of $\calP$. 
For differentiation, we often refer to the vertices of $\calP$ as {\em obstacle vertices} and the edges of $\calP$ as {\em obstacle edges}.
%We use {\em obstacle vertex} to refer to $s$ or a vertex of $\calP$ and use {\em obstacle edge} to refer to an edge of an obstacle of $\calP$.

For any point $p$ in $\calP$, let $\vis(p)$ denote the {\em visibility polygon} of $p$, i.e., $\vis(p)$ consists of all points $q\in \calP$ such that $\overline{pq}\subseteq \calP$. The size of $\vis(p)$ is $O(n)$~\cite{ref:HeffernanAn95}. 
We further define $\vis_{\calB}(p)=\vis(p)\cap \calB$, i.e., the portions of $\calB$ that are visible to $p$, which consists of $O(n)$ segments on $\calB$.

% \begin{figure}[t]
% \begin{minipage}[t]{\textwidth}
% \begin{center}
% \includegraphics[height=2.5in]{spm.pdf}
% \caption{ Illustrating the shortest path map $\spm(s)$. The solid (red) curves are walls and the
% (blue) dotted segments are windows. The anchor of each cell is also shown with a black point.}
% \label{fig:spm}
% \end{center}
% \end{minipage}
% \vspace{-0.15in}
% \end{figure}

\paragraph{Shortest path maps.}
The {\em shortest path map} $\spm(s)$ of a point $s\in \calP$ is a decomposition of $\calP$ into cells such that all points $t$ in the cell $\sigma$ have the same anchor in any shortest \st\ path~\cite{ref:HershbergerAn99,ref:MitchellA91}; the anchor is called the {\em root} of $\sigma$.  
Each edge of $\spm(s)$ is an obstacle edge fragment, an {\em extension segment} (i.e., extended from an obstacle vertex $u$ in the direction away from $v$, where $v$ is the anchor of $u$ in a shortest $s$-$u$ path), or a {\em bisector curve}
which is the locus of points $p$ with $d(s,u)+|{pu}|=d(s,v)+|{pv}|$ for two obstacle vertices $u$ and $v$. 
See Fig.~\ref{fig:spm}. Although an extension segment can be viewed as a degenerated bisector curve, we use bisector curve to refer to the general case only. As such, any point on a bisector curve has two topologically different paths from $s$. 
%Each bisector curve is in general a hyperbola; a special case happens if one of $u$ and $v$ is the anchor of the other, in which case their bisector curve is a straight line.  
%called {\em extension segment} (because it is a segment extending from $v$ along the direction from $u$ to $v$, assuming $u$ is the anchor of $v$). From now on, we use {\em bisector curve} to only refer to a general hyperbola that is not an extension segment. As such, each edge of $\spm(s)$ is an obstacle edge fragment, a bisector curve, or an extension segment. 
%Following the notation in~\cite{ref:Eriksson-BiqueGe15}, we differentiate between two types of bisector curves: {\em walls} and {\em windows}. A bisector curve of $\spm(s)$ is a {\em wall} if there exist two topologically different shortest paths from $s$ to each point of the edge; otherwise (i.e., the above special case) it is a {\em window} (e.g., see Fig.~\ref{fig:spm}). 
Following the notation in~\cite{ref:Eriksson-BiqueGe15}, we call the intersection of two bisector curves a {\em triple point}, which has at least three topologically different shortest paths from $s$.  
%meaning that it is the intersection of at least two bisector curves. 
%As such, each edge of $\spm(s)$ is an obstacle edge fragment, a wall, or a window. 

It is known that $\spm(s)$ is of size $O(n)$~\cite{ref:HershbergerAn99,ref:MitchellSh96} and can be computed in $O(n\log n)$ time~\cite{ref:HershbergerAn99} or in $O(n+h\log h)$ time after $\calP$ is triangulated~\cite{ref:WangA23} (note that triangulating $\calP$ can be done in $O(n+h\log h)$ time by a recent algorithm of Chan~\cite{ref:ChanTr26}). 
%(note that triangulating $\calP$ can be done in $O(n\log n)$ time or in $O(n+h\log^{1+\epsilon}h)$ time~\cite{ref:Bar-YehudaTr94}). 

%For any subset $R$ of $\calP$ (e.g., $R=\calB$), the portion of $\spm(s)$ restricted to $A$ is denoted by  $\spm_A(s)$, i.e., $\spm_A(s)=\spm(s)\cap A$. For example, 
Recall that $\spmb(s)$ refers to the portion of $\spm(s)$ restricted to $\calB$, the boundary of $\calP$. Hence, $\spmb(s)$ consists of a set of segments each of which belongs to a cell of $\spm(s)$.

\paragraph{Shortest path trees (SPT) and SPT-equivalence decompositions.}
For a point $s\in \calP$, the {\em shortest path tree}, denoted by $\spt(s)$, is a spanning tree of $s$ and the vertices of $\calP$ such that the path in the tree from $s$ to any vertex is a shortest path in $\calP$. 
%By our general position assumption, $\spt(u)$ is unique for any obstacle vertex $u$. 
Given $\spm(s)$, $\spt(s)$ can be easily obtained in $O(n)$ time. 

The {\em SPT-equivalence decomposition} of $\calP$, denoted by $\Psispt$, is the subdivision of $\calP$ into cells such that $\spt(s)$ is combinatorially the same for all points $s$ in the same cell. Note that $\Psispt$ can be obtained by overlaying the shortest path maps $\spm(u)$ for all obstacle vertices $u$ of $\calP$~\cite{ref:ChiangTw99}. 
Since each shortest path map is of size $O(n)$, the worst-case complexity of $\Psispt$ is $O(n^4)$ and this bound is tight~\cite{ref:ChiangTw99}. 
%Also, it is not difficult to see that $\Psispt$ is a refinement of $\Psivis$, i.e., each cell of $\Psispt$ is contained in a cell of $\Psivis$, because if two points have topologically different visibility polygons, then they must have different shortest path trees. 

%In general, let $A$ be a subset of $\calP$, we use $\Psispt_A$ to denote the subdivision of $A$ into cells so that $\spt(s)$ is combinatorially the same for all points $s$ in the same cell, i.e., $\Psispt_A=\Psispt\cap A$. 

Define $\Psisptb$ as the subdivision of $\calB$ into segments so that $\spt(s)$ is combinatorially the same for all points $s$ in the same segment. It is not difficult to see that $\Psisptb$ can be obtained by overlapping the segments of $\spmb(u)$ for all obstacle vertices $u$. As each $\spmb(u)$ has $O(n)$ segments, $\Psisptb$ has $O(n^2)$ segments. This bound is also tight in the worst case. As such, $\Psisptb$ can be computed in $O(n^2\log n)$ time.

\paragraph{Shortest path map (SPM) equivalence decompositions.}
We follow the definitions of $\Psibb$, $\Psibp$, $\Psipb$, and $\Psipp$ (which is also $\Psi$) in Section~\ref{sec:intro}. As all obstacle vertices are on $\calB$, both $\Psibb$ and $\Psibp$ are refinements of $\Psisptb$ and thus their sizes are $\Omega(n^2)$ in the worst case~\cite{ref:ChiangTw99}. Similarly, both $\Psipb$ and $\Psipp$ are refinements of $\Psispt$ and thus their sizes are $\Omega(n^4)$ in the worst case.

\paragraph{General position assumption.}
For ease of exposition, we make the following general position assumption for $\calP$: 
%(1) No three obstacle vertices are colinear; 
(1) For any obstacle vertex $s$, $\spm(s)$ does not have a vertex of degree larger than three; 
(2) for any point $s\in \calB$, $\spm(s)$ does not have a vertex of degree larger than four; 
(3) for any point $s\in \calP\setminus \calB$, $\spm(s)$ does not have a vertex of degree larger than five; 
%(1) For any obstacle vertex $v$, there does not exist another obstacle vertex $s$ and two distinct obstacle vertices $u_i$, $1\leq i\leq 2$, so that $|su_1|+d(u_1,v)=|su_2|+d(u_2,v)$; 
%(2) for any obstacle vertex $v$, there does not exist a point $s\in \calB$ and three distinct obstacle vertices $u_i$, $1\leq i\leq 3$, all visible to $s$, so that $|su_i|+d(u_i,v)$ is the same for all $i=1,2,3$; 
%(3) for any obstacle vertex $v$, there does not exist a point $s$ in the interior of $\calP$ and four distinct obstacle vertices $u_i$, $1\leq i\leq 4$, all visible to $s$, so that $|su_i|+d(u_i,v)$ is the same for all $1\leq i\leq 4$. 
%(note that this assumption is consist with the usual assumption that for each obstacle vertex $v$, for any point $s\in $); 
(4) every pair of obstacle vertices has a unique shortest path;
(5) no three obstacle vertices are collinear.

\section{The decomposition $\boldsymbol{\Psibb}$}
\label{sec:psibb}

Recall that $\Psibb$ is the decomposition of $\calB$ into segments such that $\spmb(s)$'s for all points $s$ in the same segment are topologically equivalent. 
%In the following, we first present an output-sensitive algorithm to compute $\Psibb$ and then prove an $O(n^{4+\epsilon})$ upper bound for $|\Psibb|$. 

%\subsection{Algorithm for computing $\boldsymbol{\Psibb}$}
%\label{sec:algopsibb}
Recall that $\Psibb$ is a refinement of $\Psisptb$. 
We first compute $\Psisptb$, which takes $O(n^2\log n)$ time as discussed in Section~\ref{sec:pre}.

%Let $e_s=\overline{ab}$ be a segment of $\Psisptb$. Without loss of generality, we assume that $a$ is the origin and $b$ is at the $x$-axis and to the right of $a$. 

\paragraph{Event points.}
Consider a segment $e_s$ of $\Psisptb$. 
Suppose we move a point $s$ on $e_s$ from one endpoint to the other. 
$\spmb(s)$ will change topologically in the following two situations; we call $s$ an {\em event point} when either situation happens. 
%Since $e_s$ is an edge of $\Psisptb$, $\spmb(s)$ will change topologically when $s$ is crossing such a point $s'\in e_s$ as follows (we call $s'$ an {\em event point} for $s$). 
%(1) A wall endpoint on an obstacle edge $e$ meets an endpoint of $e$. 
\begin{enumerate}
    \item At the event point, two bisector curves of $\spm(s)$ have their endpoints on $\calB$ meet at an interior point of an obstacle edge $e$. Right before $s$ reaches the event point, the two bisector curves have their endpoints both on $e$ (i.e., as $s$ moves towards the event point, the two bisector endpoints move closer on $e$). See Fig.~\ref{fig:bbevent1}, where (a), (b), and (c) illustrate the scenarios when $s$ is before, at, and after the event point, respectively. 
    \item This event case can be considered the inverse of the first event case.
    At the event point, two bisector curves of $\spm(s)$ have their endpoints on $\calB$ meet at an interior point of an obstacle edge $e$. Right before $s$ reaches the event point, the two bisector curves intersect at a triple point $q$ in the interior of $\calP$ (i.e., there is a third bisector curve $\gamma$ whose endpoint is $q$ and whose other endpoint is at a point $p\in e$, and    
    as $s$ moves towards the event point, $q$ moves on $\gamma$ towards $p$). See Fig.~\ref{fig:bbevent1}, where (c), (b), and (a) illustrate the scenarios when $s$ is before, at, and after the event point, respectively. 
\end{enumerate}

As such, when $s$ moves between two adjacent event points on $e_s$, $\spmb(s)$ does not change combinatorially. Hence, to construct $\Psibb$, it suffices to compute all event points on all segments $e_s$ of $\Psisptb$. 

\begin{figure}[t]
\begin{minipage}[t]{\textwidth}
\begin{center}
\includegraphics[height=0.8in]{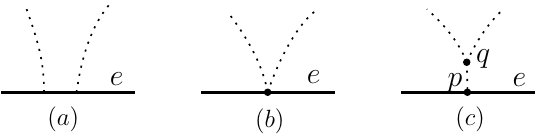}
\caption{ Illustrating the combinatorial change of $\spmb(s)$.}
\label{fig:bbevent1}
\end{center}
\end{minipage}
\vspace{-0.1in}
\end{figure}

As a ``warm-up,'' we prove an $O(n^{5})$ upper bound by using a simple brute-force method; similar methods were used for other problems in the previous work, e.g., \cite{ref:BaeTh13,ref:ChiangTw99,ref:WangOn18}. 

\paragraph{An $\boldsymbol{O(n^5)}$ upper bound for $\boldsymbol{|\Psibb|}$.}
Each edge $e$ of $\calB$ is processed as follows. For any subset $V$ of three (not necessarily distinct) obstacle vertices $u_i$, $1\leq i\leq 3$, let $\calO_V$ be the overlay of $\spmb(u_i)$, $1\leq i\leq 3$, which comprises $O(n)$ segments. For each segment $e_t\in \calO_V$, let $v_i$ be the root of the cell of $\spm(u_i)$ that contains $e_t$, $1\leq i\leq 3$. We solve the equation $d_{u_1,v_1}(s,t)=d_{u_2,v_2}(s,t)=d_{u_3,v_3}(s,t)$ for $s\in e$ and $t\in e_t$. Since the equation has $O(1)$ solutions, there are $O(1)$ pairs  $(s,t)$ with $s\in e$ and $t\in e_t$. As there are $O(n)$ segments in $\calO_V$, there are $O(n)$ such pairs $(s,t)$ with $s\in e$ and $t\in \calB$ for $V$. Enumerating all subsets $V$ of three obstacle vertices gives an upper bound $O(n^4)$ for all such pairs $(s,t)$ with $s\in e$. Processing all obstacle edges $e$ as above gives an $O(n^5)$ upper bound for all such pairs $(s,t)$ with $s\in \calB$. 

Consider any event point $s$ of $\Psibb$, where $\spmb(s)$ makes combinatorial changes. Let $e$ be the obstacle edge that contains $s$. Then, there must be a point $t\in \calB$ so that $(s,t)$ is a solution to an equation for a subset $V$ of three obstacle vertices as above. As such, the number of event points of $\Psibb$ must be $O(n^5)$. 

The above proof can be easily turned into an algorithm to compute all event points and thus $\Psibb$ in $O(n^5\log n)$ time. 
\medskip

%\paragraph{Improvement.}
In the rest of this section, we present an improved $O(n^{4+\epsilon})$ time algorithm to compute $\Psibb$. Our approach also establishes an upper bound $O(n^{4+\epsilon})$ for $|\Psibb|$.
%Using similar techniques, we will also prove the upper bounds for $\Psi_{\xi}(\calB)$ and $\Psi_{B_s}(B_t)$.

\subsection{Defining functions}

For each obstacle vertex $u$, we define a function $f_u$ as follows. 

\begin{definition}\label{def:fun}
For any obstacle vertex $u$, for any point $s\in \visb(u)$ and any point 
$t\in \calB$, define $f_u(s,t)=d_{u,v_t}(s,t)$, where $v_t$ is the root of the cell of $\spm(u)$ that contains $t$.    
\end{definition}

Recall that $\visb(u)$ is the portion of the visibility polygon $\vis(u)$ of $u$ on $\calB$. For simplicity, we assume $f_u(s,t)=\infty$ if $s$ is not visible to $u$. 

If we use the positions of $s$ and $t$ on $\calB$ as variables, then $f_u(s,t)$ is a bivariate piecewise algebraic function of constant degree (but not necessarily of constant complexity). Also, $f_u(s,t)$ defines a graph that is a surface patch in $\bbR^3$ in which the first two coordinates correspond to the positions of $s$ and $t$ on $\calB$ and the third coordinate is $f_u(s,t)$. 
A convenient way to think about this is to place all edges of $\calB$ on a line, in which $s$ and $t$ change. More specifically, obstacles can be arbitrarily ordered, but edges of the same obstacle of $\calP$ should be placed consecutively following their order along the obstacle boundary (refer to \cite{ref:BaeQu12} for a detailed treatment of this; in \cite{ref:BaeQu12}, the authors differentiated the point $s\in \calB$ from its ``coordinate'' on $\calB$ with a different notation; here to simplify the notation, we use $s$ to denote both a point on $\calB$ and its coordinate). 
We refer to the copy of the line for $s$ (resp., $t$) as the {\em $s$-axis} (resp., {\em $t$-axis}); they together form the $st$-plane. We also use $f_u(s,t)$ to refer to the graph or the surface patch defined by it. 
%For any point $p$ in the graph, let $\eta_s(p)$ (resp., $\eta_t(p)$) denote the projection of $p$ onto the $s$-axis (resp., $t$-axis). 

For the function $f_u(s,t)$, while the domain of $t$ is the entire $\calB$, the domain of $s$ comprises the $O(n)$ segments of $\visb(u)$. We call the endpoints of the segments of $\visb(u)$ the {\em breakpoints induced by $u$} on the $s$-axis. Also, for each segment of $\spm_{\calB}(u)$, the vertex $v_t$ (as in Definition~\ref{def:fun}) is the same for all points $t$ on the segment. We call the endpoints of the segments of $\spm_{\calB}(u)$ the {\em breakpoints induced by $u$} on the $t$-axis. For each segment $\sigma_s\in \visb(u)$ and each segment $\sigma_t\in \spmb(u)$, $f_u(s,t)=d_{u,v}(s,t)$ for any $s\in \sigma_s$ and $t\in \sigma_t$, where $v$ is the root of the cell of $\spm(u)$ containing $\sigma_t$, and thus $f_u(s,t)$ is of constant complexity on $s\in \sigma_s$ and $t\in \sigma_t$; we refer to it as an {\em elementary function}. As such, since both $\visb(u)$ and $\spmb(u)$ have $O(n)$ segments, $f_u(s,t)$ can be decomposed into $O(n^2)$ elementary functions each of which is defined on $s\in \sigma_s$ and $t\in \sigma_t$ for 
a segment $\sigma_s$ of $\visb(u)$ and a segment $\sigma_t$ of $\spmb(u)$. 
%by breakpoints induced by $u$ on the $s$-axis and $t$-axis. 
%As such, we can partition the $st$-plane into a grid of size $O(n^2)$ by $O(n)$ axix-parallel lines so that $f_u(s,t)$ on each grid cell is an elementary function. 
For differentiation, we refer to the original function $f_u(s,t)$ for $s\in \visb(u)$ and $t\in \calB$ as a {\em compound function}. 

%The observation below follows from the definition of $f_u$. 
%\begin{observation}\label{obser:superfun}
%For any two points $s',t'\in \calB$, the surface patch in $\bbR^3$ defined by $f_u(s,t)$ has at most one point whose $s$-coordinate and $t$-coordinate are equal to $s'$ and $t'$, respectively. 
%\end{observation}

% \paragraph{Remark.}
% Although Observation~\ref{obser:superfun} looks simple, our way of defining functions is crucial to the success of our approach. Indeed, as will be clear later, the observation makes it possible to adapt the random sampling analysis technique of \cite{ref:HalperinNe94,ref:SharirAl94,ref:SharirDa95} to our problem setting. 
% \medskip

Let $F$ denote the set of compound functions $f_u(s,t)$ for all obstacle vertices $u$. 
For convenience, we also consider $F$ the set of elementary functions of all these compound functions. 
%an elementary function of a compound function $f_u(s,t)$ is also considered an elementary function of $F$. 
As such, $F$ has $O(n^3)$ elementary functions. 

\paragraph{Lower envelope.}
%Define $\calA(F)$ to be the arrangement of the all surface paths defined by the functions of $F$.
Define $\calL(F)$ to be the lower envelope of all functions of $F$. 
%Let $V(F)$ denote the set of vertices in the interior of $\calL(F)$. 
By definition, for two points $s,t\in \calB$, there is a unique point $p\in \calL(F)$ whose projection onto the $st$-plane is $(s,t)$. Further, if $s$ is not visible to $t$, then the third coordinate of $p$ must be equal to $d(s,t)$, because a shortest \st\ path must contain an obstacle vertex. 

Due to our general position assumption, the common intersection of three elementary functions of $F$ is a vertex of $\calL(F)$. 
%Clearly, each vertex of $\calL(F)$ is the common intersection of at least three elementary functions of $F$ (actually exactly three due to our general position assumption). 
%(in fact, our general position assumption guarantees that each vertex of $\calL(F)$ is the common intersection of exactly three elementary functions). 
%(we could make the general position assumption stronger so that every vertex of $\calA(F)$ is defined by exactly three elementary functions). 
The following observation follows from our general position assumption. 
%which implies that every three elementary functions of $F$ can define $O(1)$ vertices of $\calA(F)$. 

%\begin{observation}\label{obser:vertexfun}
%Each interior vertex of $\calL(F)$ is defined by exactly three elementary functions of $F$. 
%\end{observation}

%For any two points $s,t\in\calB$, unless $s$ is visible to $t$, the $z$-coordinate of the unique point of $\calL(F)$ with the first two coordinates $(s,t)$ gives $d(s,t)$. 

\begin{observation}\label{obser:threefun}
Every three elementary functions of $F$ have $O(1)$ common intersection points. 
\end{observation}
\begin{proof}
Consider three elementary functions $f_{u_i}(s,t)=d_{u_i,v_i}(s,t)=|su_i|+d(u_i,v_i)+|v_it|$, for $i=1,2,3$. A common intersection point of these three functions corresponds to a pair of points $(s,t)$ that is a solution of the equation $d_{u_1,v_1}(s,t)=d_{u_2,v_2}(s,t)=d_{u_3,v_3}(s,t)$. Since these are three different elementary functions, for any $i,j\in\{1,2,3\}$ with $i\neq j$, either $u_i\neq u_j$ or $v_i\neq v_j$. Due to our general position assumption, the equation $d_{u_1,v_1}(s,t)=d_{u_2,v_2}(s,t)=d_{u_3,v_3}(s,t)$ has $O(1)$ solutions. The observation thus follows. 
\end{proof}

The following observation reduces the problem of bounding $|\Psibb|$ to bounding the combinatorial complexity of $\calL(F)$, by showing that every event point of $\Psibb$ corresponds to a distinct vertex of $\calL(F)$.  

\begin{observation}\label{obser:eventvertex}
Every event point of $\Psibb$ corresponds to a distinct vertex of $\calL(F)$, and 
$|\Psibb|=O(n^2+|\calL(F)|)$.
\end{observation}
\begin{proof}
Recall that $\Psibb$ consists of vertices of $\Psisptb$ plus event points 
%(i.e., those in our algorithm for computing $\Psibb$) 
in the interior of segments of $\Psisptb$. Since $\Psisptb$ has $O(n^2)$ vertices, to prove $|\Psibb|=O(n^2+|\calL(F)|)$, it suffices to show that the number of event points is at most $|\calL(F)|$. To this end, we show that every event point of $\Psibb$ corresponds to a distinct vertex of $\calL(F)$. 

Consider such an event point $s\in \Psibb$. By definition, there is a point $t\in \calB$ such that there are three topologically different shortest $s$-$t$ paths $\overline{su_i}\cup \pi(u_i,v_i)\cup \overline{v_it}$ with $i=1,2,3$. Hence, $(s,t,d(s,t))$ must be a point in $\bbR^3$ on the lower envelope $\calL(F)$ and is also a common intersection of three elementary functions of $F$: $d_{u_i,v_i}(s,t)$, $1\leq i\leq 3$. As such, $(s,t,d(s,t))$ must be a vertex of $\calL(F)$. 

Note that no other event point corresponds to the same vertex $(s,t,d(s,t))$. Indeed, to have another event point $s'$ also correspond to $(s,t,d(s,t))$, $s=s'$ must hold. 
%Since $s'\neq s$, we know that $s'$ cannot correspond to $(s,t,d(s,t))$. 

We also remark that it is possible that there is another point $t'\in \calB$ such that $t'\neq t$ and there are three topologically different shortest $s$-$t'$ paths. In that case, $s$ corresponding to another vertex $(s,t',d(s,t'))$ of $\calL(F)$. Hence, this observation also holds if an event point $s$ is caused by multiple topological changes of $\spmb(s)$, in which case $s$ is considered multiple event points that are co-located. 
\end{proof}

In what follows, we focus on computing the vertices of $\calL(F)$. As $F$ has $O(n^3)$ elementary bivariate functions, applying the result of \cite{ref:HalperinNe94,ref:SharirAl94,ref:SharirDa95} directly can give $|\calL(F)|=O(n^{6+\epsilon})$ and compute all vertices in $O(n^{6+\epsilon})$. In the following, we present an improved algorithm of $O(n^{4+\epsilon})$ time and prove the bound $|\calL(F)|=O(n^{4+\epsilon})$. 

\subsection{Computing the vertices of $\boldsymbol{\calL(F)}$}

For convenience, when considering functions $f_{u}(s,t)$ of $F$, we use $\calB_s$ (resp., $\calB_t$) to refer to the copy of $\calB$ for $s$ (resp., $t$), i.e., $s$ changes on $\calB_s$ while $t$ changes on $\calB_t$. 

For a line segment $e_s$ on an obstacle edge of $\calB_s$ and a line segment $e_t$ on an obstacle edge of $\calB_t$, let $\calL(F)[e_s,e_t]$ denote the portion of $\calL(F)$ of the functions $f_u(s,t)\in F$ restricted to $s\in e_s$ and $t\in e_t$. We will partition $\calB_t$ into $O(n)$ line segments $e_t$ and present an algorithm to compute all vertices of $\calL(F)[\calB_s,e_t]$ in $O(n^{3+\epsilon})$ time. Doing this for all $O(n)$ such segments $e_t$ of $\calB_t$ will compute all vertices of $\calL(F)$ in $O(n^{4+\epsilon})$ time. To compute the vertices of $\calL(F)[\calB_s,e_t]$, we partition $\calB_s$ into $O(n)$ line segments $e_s$ and present an algorithm to compute all vertices of $\calL(F)[e_s,e_t]$ in $O(n^{2+\epsilon})$ time. Doing this for all $O(n)$ such segments $e_s$ will compute all vertices of $\calL(F)[\calB_s,e_t]$ in $O(n^{3+\epsilon})$ time. 
%The details are discussed in the following. 
%In what follows, we present the details on how to make the partitions on $\calB_s$ and $\calB_t$, and the corresponding algorithms. 

%Consider an obstacle edge $e_t\subseteq \calB_t$. The breakpoints partition $e_t$ into multiple {\em elementary segments}. We know that $\calB_t$ has $O(n^2)$ elementary segments. We assume that $e_t$ has more than $n$ 
\paragraph{The partition of $\calB_t$.}
Recall that there are $O(n^2)$ breakpoints on $\calB_t$, which partition $\calB_t$ into $O(n^2)$ line segments, and we call them {\em elementary intervals}. For each obstacle edge $e\subseteq \calB_t$, starting from an endpoint of $e$, we group every $n$ adjacent elementary intervals on $e$ into a {\em super interval}, except that the last super interval may contain less than $n$ elementary intervals. As $\calB_t$ has $O(n^2)$ elementary intervals and $O(n)$ obstacle edges, $\calB_t$ is partitioned into $O(n)$ super intervals each of which consists of at most $n$ elementary intervals. This is our partition for $\calB_t$. Note that all super intervals can be computed in $O(n^2)$ time, assuming that the shortest path maps of all obstacle vertices are available. Since the shortest path map of each obstacle vertex can be computed in $O(n\log n)$ time~\cite{ref:HershbergerAn99}, we conclude that all super intervals of $\calB_t$ can be obtained in $O(n^2\log n)$ time. 

Let $e_t$ be a super interval of $\calB_t$. 
We next show that computing $\calL(F)[\calB_s,e_t]$ can be done in $O(n^{3+\epsilon})$ time. To this end, we first introduce our partition for $\calB_s$ with respect to $e_t$. 

\paragraph{The partition of $\calB_s$.}
Recall that $e_t$ contains at most $n$ elementary intervals. For notational convenience, we assume that $e_t$ contains exactly $n$ elementary intervals. Then, $e_t$ has $n-1$ breakpoints in its interior (these breakpoints partition $e_t$ into $n$ elementary intervals). For each obstacle vertex $u$, let $k_u$ denote the number of breakpoints of $e_t$ induced by $u$. As such, we have $\sum_{u\in V}k_u=n-1$ (recall that $V$ is the set of all obstacle vertices of $\calP$). 

Recall that $\calB_s$ has $O(n^2)$ breakpoints, which partition $\calB_s$ into $O(n^2)$ line segments, and we also call them {\em elementary intervals} for $s$. For each breakpoint $p$, we define $k_p=k_u$, where $u$ is the obstacle vertex that induces $p$. Consider an obstacle edge $e\subseteq \calB_s$. Let $a$ and $b$ be the two endpoints of $e$, respectively. Starting from the elementary interval containing $a$, we group adjacent elementary intervals on $e$ as many as possible into a {\em super interval} $\xi$ until one of the following two cases to happen (the last elementary interval that causes one of these cases happens is included in $\xi$): 
%(1) the super interval contains $n$ elementary intervals; 
(1) the sum of $k_p$ of all breakpoints in the interior of $\xi$ is larger than $2n$; (2) $\xi$ contains $b$. 
If $\xi$ does not contain $b$, then starting from the next elementary interval on $e$ after $\xi$, we create another super interval. We continue this until $b$ is included in last super interval. We create super intervals on other obstacle edges of $\calB_s$ in the same way. 
Since $k_p\leq n-1$ for each breakpoint $p$, we have the following observation. 

\begin{lemma}\label{lem:superinterval}
    For each super interval of $\calB_s$, the sum of $k_p$ of all breakpoints $p$ in its interior is no more than $3n$. 
\end{lemma}
\begin{proof}
Consider a super interval $\xi$. If $\xi$ consists of only one elementary interval, then the observation trivially follows as there is no breakpoint in the interior of $\xi$. Otherwise, let $\xi'\subseteq \xi$ be the sub-interval of $\xi$ excluding the last elementary interval, denoted by $\xi''$. Then, $\xi'$ and $\xi''$ are separated by a breakpoint $p'$. By definition, the sum of $k_p$ of all breakpoints $p$ in the interior of $\xi'$ must be no more than $2n$. Since $k_{p'}\leq n-1$ and the sum of $k_p$ of all breakpoints $p$ in the interior of $\xi$ is equal to $k_{p'}$ plus the sum of $k_p$ of all breakpoints $p$ in the interior of $\xi'$, we obtain that  the sum of $k_p$ of all breakpoints $p$ in the interior of $\xi$ is at most $3n$.
\end{proof}

The following lemma provides an upper bound on the number of super intervals of $\calB_s$. 

\begin{lemma}
    The number of super intervals of $\calB_s$ is $O(n)$. 
\end{lemma}
\begin{proof}
    By definition, there are two cases for a super interval to be created. 
    For the second case, observe that each obstacle edge can have at most one super interval under the second case. Hence, the number of second case super intervals is $O(n)$. 

    We now discuss the first case. For each obstacle vertex $u$, let $P_u$ denote the set of breakpoints induced by $u$. Note that $|P_u|=O(n)$. Recall that $\sum_{u\in V}k_u=n-1$. Hence, the sum of $k_p$ of all breakpoints $p$ on $\calB_s$ is $\sum_{u\in V}\sum_{p\in P_u} k_p = \sum_{u\in V}\sum_{p\in P_u} k_u =O(n)\cdot \sum_{u\in V}k_u=O(n^2)$. For each super interval $\xi$ in the first case, by definition, the total sum of $k_p$ of all breakpoints $p$ in the interior of $\xi$ is larger than $2n$. Therefore, the total number of first case super intervals on $\calB_s$ is $O(n)$. The lemma thus follows. 
\end{proof}

It is not difficult to see that the super intervals of $\calB_s$ can be easily computed in $O(n^2)$ time, assuming that all breakpoints and their induced vertices are already available. The breakpoints can be obtained by computing the visibility polygons $\vis(u)$ of all obstacle vertices $u\in V$. As computing $\vis(u)$ for each $u$ can be done in $O(n\log n)$ time~\cite{ref:HeffernanAn95}, we conclude that all super intervals of $\calB_s$ can be computed in $O(n^2\log n)$ time. 

\begin{lemma}
    For any super interval $e_s$ of $\calB_s$, $\calL(F)[e_s,e_t]$ has $O(n^{2+\epsilon})$ vertices and all vertices can be computed in $O(n^{2+\epsilon})$ time. 
\end{lemma}
\begin{proof}
Recall that $e_s$ lies on a single obstacle edge, so does $e_t$. 
To prove the lemma, we will show that the number of elementary functions $f_u(s,t)$ defined on $s\in e_s$ and $t\in e_t$ is $O(n)$. Recall that each elementary function is a bivariate algebraic function of constant degree and constant size. Consequently, applying the result of \cite{ref:HalperinNe94,ref:SharirAl94,ref:SharirDa95} on these $O(n)$ elementary functions can give $|\calL(F)[e_s,e_t]|=O(n^{2+\epsilon})$ and compute all vertices of $\calL(F)[e_s,e_t]$ in $O(n^{2+\epsilon})$ time. 

We now argue that the number of elementary functions $f_u(s,t)$ defined on $s\in e_s$ and $e_t$ is $O(n)$. For any subsets $e_s'\subseteq e_s$ and $e_t'\subseteq e_t$, let $F[e_s',e_t']$ denote the set of elementary functions $f_u(s,t)$ defined on $s\in e_s'$ and $t\in e_t'$. Our goal is to prove $|F[e_s,e_t]|=O(n)$. 

We order the elementary intervals of $e_s$ from one end to the other. We do the same for $e_t$. 
%Recall that $e_t$ has at most $n$ elementary intervals. 
Consider the first elementary interval $I_s$ of $e_s$ and the first elementary interval $I_t$ of $e_t$. Since $I_s$ is an elementary interval of $\calB_s$ and $I_t$ is an elementary interval of $\calB_t$, we have $|F[I_s,I_t]|\leq n$ as there are $n$ obstacle vertices. Let $m=|F[I_s,I_t]|$. Let $I_t'$ be the elementary interval of $e_t$ adjacent to $I_t$. By definition, $I_t$ and $I_t'$ are two line segments on the same obstacle edge separated by a breakpoint of $\calB_t$. Hence, comparing $F[I_s,I_t']$ to $F[I_s,I_t]$, at most one function of $F[I_s,I_t]$ is changed to a new function in $F[I_s,I_t']$ while all other functions of $F[I_s,I_t]$ are still in $F[I_s,I_t']$ (i.e., each of these functions is originally defined on $t\in I_t$ but we can extend its domain to $t\in I_t\cup I_t'$). As such, $F[I_s,I_t\cup I_t']$ has $m+1$ elementary functions. If we continue the same analysis for the remaining elementary intervals of $e_t$ one by one, we can obtain that $F[I_s,e_t]$ has at most $m+n\leq 2n$ elementary functions since $e_t$ has at most $n$ elementary intervals. 

Now consider the elementary interval $I_s'$ of $e_s$ adjacent to $I_s$. Let $p$ be the breakpoint that separates $I_s$ and $I_s'$. We can follow a similar analysis to the above. When we move from $I_s$ to $I_s'$, at most $k_p$ functions are changed from $F[I_s,e_t]$ to $F[I_s',e_t]$ while the other functions of $F[I_s,e_t]$ are still in $F[I_s',e_t]$ (each of them is originally defined on $s\in I_s$ but we can extend their domain to $s\in I_s\cup I_s'$). Hence, the number of functions in $F[I_s\cup I_s',e_t]$ is at most $k_p+|F[I_s,e_t]|\leq k_p+2n$. If we continue the same analysis for the rest of the elementary intervals of $e_s$, then since the sum of $k_p$ for all breakpoints $p$ in the interior of $e_s$ is at most $3n$ by Lemma~\ref{lem:superinterval}, we can obtain that $|F[e_s,e_t]|=O(n)$.

Note that following the above analysis, we can easily find all functions of $F[e_s,e_t]$ in $O(n)$ time, after which we can apply the algorithm of \cite{ref:HalperinNe94,ref:SharirAl94,ref:SharirDa95} to compute all vertices of $\calL(F)[e_s,e_t]$ in $O(n^{2+\epsilon})$ time. The lemma thus follows. 
\end{proof}

With the above lemma, since $\calB_s$ has $O(n)$ super intervals, it follows that $\calL(F)[\calB_s,e_t]$ has $O(n^{3+\epsilon})$ vertices and they can be computed in $O(n^{3+\epsilon})$ time. Furthermore, as $\calB_t$ has $O(n)$ super intervals $e_t$, we conclude that $\calL(F)[\calB_s,\calB_t]$, which is $\calL(F)$, has $O(n^{4+\epsilon})$ vertices and they can be computed in $O(n^{4+\epsilon})$ time. We thus have the following theorem. 

\begin{theorem}
The lower envelope $\calL(F)$ has $O(n^{4+\epsilon})$ vertices and all these vertices can be computed in $O(n^{4+\epsilon})$ time.
\end{theorem}

Following our earlier discussion and Observation~\ref{obser:eventvertex}, we have the following result. 

\begin{corollary}\label{coro:psibb}
The combinatorial complexity of $\Psibb$ is $O(n^{4+\epsilon})$ and  $\Psibb$ can be computed in $O(n^{4+\epsilon})$ time. 
\end{corollary}

\section{The decomposition $\boldsymbol{\Psibp}$}
\label{sec:psibp}

In this section, we discuss the decomposition $\Psibp$. In particular, we will use $\Psibp$ to handle two-point shortest path queries with $s\in \calB$ and $t\in \calP$ in Section~\ref{sec:query}. 
We first present an algorithm to compute $\Psibp$ in $O(|\Psibp|\cdot \log n)$ time and then prove an $O(n^5)$ upper bound for its combinatorial complexity $|\Psibp|$. 

\subsection{Algorithm for computing $\boldsymbol{\Psibp}$}
%We follow a similar algorithm to that for $\Psibb$. 
Recall that $\Psibp$ is a refinement of $\Psibb$, which is a refinement of $\Psisptb$. We start with computing $\Psisptb$, which takes $O(n^2\log n)$ time. Consider an edge $e_s$ of $\Psisptb$. 

\subsubsection{Event points}
Suppose that we move $s$ on $e_s$. Since $e_s$ is a segment of $\Psisptb$, as $s$ moves on $e_s$, there are two cases where combinatorial changes on $\spm(s)$ will happen (the location of $s$ where a combinatorial change of $\spm(s)$ happens is referred to as an {\em event point} for $s$). First (called the {\em boundary case}), a combinatorial change of $\spm(s)$ may happen on $\calB$; in this case, the event point is also an event point for $\Psibb$, which is the same as discussed in Section~\ref{sec:psibb}. 
Second (called the {\em interior case}), a combinatorial change of $\spm(s)$ may happen in the interior of $\calP$. 
In this case, as discussed in~\cite{ref:ChiangTw99}, a bisector connecting two triple points contracts into a single point $t$ that has four shortest paths from $s$ (see Fig.~\ref{fig:bpevent10}). After the contraction, a new bisector is generated. 

\begin{figure}[t]
\begin{minipage}[t]{\textwidth}
\begin{center}
\includegraphics[height=1.8in]{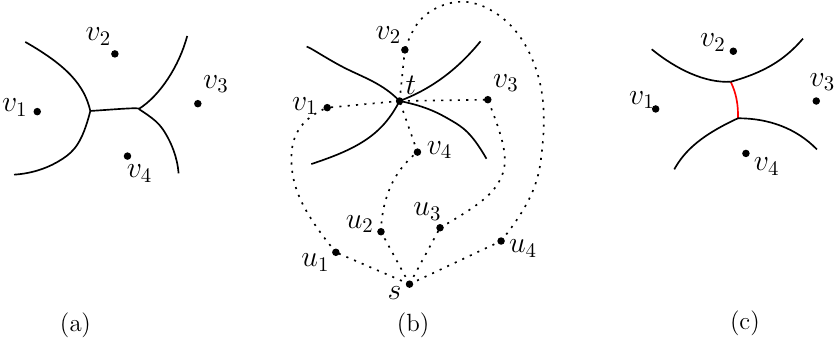}
\caption{ The solid curves are bisector curves. $v_i$, $1\leq i\leq 4$, are four obstacle vertices defining these bisector curves. (a) Before $s$ crosses an event point. (b) $s$ is at the event point: there are four shortest paths from $s$ to $t$, where a bisector contracts. (c) After $s$ passes the event point: the red curve is a new bisector curve.}
\label{fig:bpevent10}
\end{center}
\end{minipage}
\vspace{-0.1in}
\end{figure}

\subsubsection{Algorithm}
We can compute all event points on $e_s$ as follows. Suppose initially we have $\spm(s)$ available when $s$ is at one endpoint of $e_s$. As $s$ moves on $e_s$, we maintain certain information to detect the event points of the two cases.
%the following two types of information. The first type is for detecting the events in the above first case while the second type is for the second case. 

\paragraph{The boundary case.}
To compute the boundary case event points, for each bisector curve $\gamma$ of $\spm(s)$ that has an endpoint $p$ on an obstacle edge $e$, we maintain the following information. Let $u$ and $v$ be the two obstacle vertices that define $\gamma$, i.e., $u$ and $v$ are the roots of the two cells of $\spm(s)$ bounded by $\gamma$; see Fig.~\ref{fig:bbeventalgo}. 

As $s$ moves on $e_s$, $p$ will move on $e$. Note that unless $u$ and $v$ are the two endpoints of $e$, during the moving, $p$ cannot become collinear with $u$ and $v$, since otherwise the shortest path tree $\spt(s)$ would change combinatorially, contradicting that $e_s$ is an edge of $\Psisptb$. For a similar reason, $p$ will not cross either endpoint of $e$ during the moving of $s$. The position of $p$ on $e$ can be parameterized by that of $s$ on $e_s$. Let $p_1$ be a neighboring bisector curve endpoint of $p$ on $e$ (if any), and let $\gamma_1$ be the bisector curve containing $p_1$ (one of the two defining obstacle vertices of $\gamma_1$ is in $\{u,v\}$); see Fig.~\ref{fig:bbeventalgo}. As $s$ moves on $e_s$, we calculate the position of $s$ (if any) where $p_1$ will meet $p$. To this end, one needs to solve an equation of algebraic functions of constant degree to obtain $O(1)$ positions for $s$. We take the smallest value larger than the current position of $s$, i.e., this is the next position for $s$ when $p$ meets $p_1$; let $s^1_p$ denote this position. If $p$ has another neighboring bisector curve endpoint $p_2\in e$, then we calculate the corresponding position $s^2_p$ for $s$. In addition, let $p'$ be the other endpoint of $\gamma$. If $p'$ is a triple point of $\spm(s)$ (see Fig.~\ref{fig:bbeventalgo}), then $p'$ also moves as $s$ moves, but its position can also be represented a function of $s$. We calculate the next position of $s$ (if any) so that $p'$ meets $p\in e$ (i.e., the bisector curve connecting $p$ and $p'$ contracts; see the red curve $\gamma$ in Fig.~\ref{fig:bbeventalgo}); let $s_p^0$ denote the position. Among all three positions $s_p^0$, $s_p^1$, and $s_p^2$, we only maintain the closer one to the current position of $s$, denoted by $s_p$ and called a {\em candidate event} of $s$ for $p$. We add $s_p$ to an {\em event heap} $H_s$ (with the position of $s_p$ as the {\em key}). 

\begin{figure}[t]
\begin{minipage}[t]{\textwidth}
\begin{center}
\includegraphics[height=1.2in]{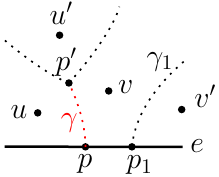}
\caption{ $u$, $v$, $u'$, and $v'$ are obstacle vertices. $p'$ is a triple point defined by $u$, $v$, and $u'$.}
\label{fig:bbeventalgo}
\end{center}
\end{minipage}
\vspace{-0.1in}
\end{figure}

\paragraph{The interior case.}
To compute the interior case event points, we main the following information. For each bisector curve of $\spm(s)$ whose both endpoints are triple points, we calculate the next position of $s$ on $e_s$ (if any) where it will contract to a single point as a candidate event point and add it to $H_s$. 

\paragraph{Main loop.}
The main loop of the algorithm works as follows. In each iteration, we extract the first candidate event $s_p$ from the event heap $H_s$ and process it as follows. First, we report $s_p$ as a ``true'' event point. Then, depending on whether the event belongs to the boundary case or the interior case, we update the events in $H_s$ related to $p$ as follows. 

In the boundary case, $s_p$ is defined by a bisector curve endpoint on $\calB$. We do the following. Depending on whether $s_p$ is $s_p^0$, there are two cases. If $s_p$ is $s_p^0$, then after $s$ moves across $s_p$, $p$ will be replaced by two new bisector endpoints on $e$ (see Fig.~\ref{fig:bbeventalgo10}). Specifically, let $v,u,w$ be the three obstacle vertices defining $p$ such that $u$ is in the middle when $s$ is at $s_p$. Then, after $s$ crosses $s_p$, the bisector curve defined by $w$ (resp., $v$) and $u$ has a new endpoint at $e$. For each such endpoint $p'$, we compute its candidate event and add it to $H_s$. Further, we update the candidate event for the neighboring bisector event of $p'$, i.e, if $p'$ has a neighboring bisector curve endpoint $p_1$ on $e$, then we delete $p_1$'s candidate event from $H_s$, and recompute a new candidate event and insert it to $H_s$. 
If $s_p$ is either $s_p^1$ or $s_p^2$, we do the following. Without loss of generality, we assume that $s_p=s_p^1$.
%, i.e., $p$ meets $p_1$ when $s$ is at $s_p$. 
Hence, when $s$ is at $s_p$, $p$ has three anchors $v,u,w$ (see Fig.~\ref{fig:bbeventalgo20}). 
%Without loss of generality, let $u$ be between in the middle of $v$ and 
%$p$ and $p_1$ have a common anchor, say, $u$. Let $v$ (resp., $w$) be the other anchor defining $p$ (resp., $p_1$). 
After $s$ moves across $s_p$, a new bisector curve arises that is defined by two of $\{v,u,w\}$, say, $v$ and $w$, and the bisector curve has $p$ as one of its endpoint. Using $v$ and $w$, we recompute the candidate event of $s$ for this new endpoint $p$ and insert it to $H_s$. Further, for each of $p$'s neighboring bisector curve endpoint on $e$, we delete its candidate from $H_s$ and recompute a new candidate event and insert it to $H_s$. 

\begin{figure}[t]
\begin{minipage}[t]{\textwidth}
\begin{center}
\includegraphics[height=1.2in]{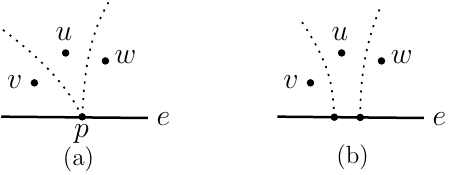}
\caption{ (a) $s$ is at $s_p$; (b) after $s$ crosses $s_p$.}
\label{fig:bbeventalgo10}
\end{center}
\end{minipage}
\vspace{-0.1in}
\end{figure}

\begin{figure}[t]
\begin{minipage}[t]{\textwidth}
\begin{center}
\includegraphics[height=1.2in]{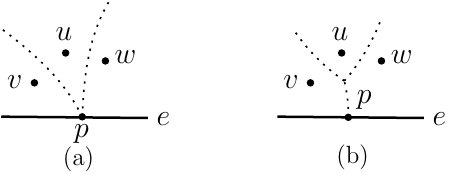}
\caption{ (a) $s$ is at $s_p$; (b) after $s$ crosses $s_p$.}
\label{fig:bbeventalgo20}
\end{center}
\end{minipage}
\vspace{-0.1in}
\end{figure}

In the interior case, $s_p$ is defined by a bisector curve of $\spm(s)$ connecting two triple points (see Fig.~\ref{fig:bpevent10}). We update $\spm(s)$ by removing the bisector curve and adding the new generated bisector. This is only a local change to $\spm(s)$ and costs $O(1)$ time since the degree of each vertex of $\spm(s)$ is at most $4$ by our general position assumption. After the update of $\spm(s)$, for each bisector curve that has one of its incident triple points updated (including the new generated bisector), we update its candidate event in $H_s$. Again, only $O(1)$ bisector curves are affected. 

The algorithm stops once $s$ reaches the other endpoint of $e_s$, at which time all event points on $e_s$ have been computed.

\paragraph{Time analysis.}
For the time analysis, first notice that $|H_s|$ is always bounded by $O(n)$ since the size of $\spm(s)$ is $O(n)$. Handling each event takes $O(\log n)$ time as each event only involves local changes on $O(1)$ bisector curves. As such, each iteration can be performed in $O(\log n)$ time. Hence, the runtime for computing all event points on $e_s$ is $O(\kappa(e_s)\cdot \log n)$, where $\kappa(e_s)$ is the number of event points on $e_s$. If we construct $\spm(s)$ initially when $s$ is at an endpoint of $e_s$, which takes $O(n\log n)$ time, then the total time for computing the event points on $e_s$ is $O(n\log n+ \kappa(e_s)\cdot \log n)$. Applying the algorithm to all segments $e_s\in\Psisptb$ will construct $\Psibp$ in total $O(n^3\log n+|\Psibp|\log n)$ time since $\Psisptb$ has $O(n^2)$ segments. 

We can further reduce the factor $O(n^3\log n)$ to $O(n^2\log n)$ as follows. Notice that the factor is dominated by the time for computing $\spm(s)$ initially for each $e_s\in \Psisptb$. Instead of doing this for each segment of $\Psisptb$, we process all segments of $\Psisptb$ on the same obstacle edge $e$ all together. Specifically, initially, we compute $\spm(s)$ when $s$ is at an endpoint of $e$. Then, we process the first segment $e_s$ of $\Psisptb$ on $e$. After $e_s$ is processed, we already have $\spm(s)$ available for the current position $s$, which is at the common endpoint of $e_s$ and the next segment $e_s'$. As such, before we process $e_s'$, we do not have to recompute the $\spm(s)$ again. In this way, we only need to compute $\spm(s)$ from scratch once for each obstacle edge, which takes $O(n^2\log n)$ time in total. Therefore, the overall runtime of the algorithm can be bounded by $O(n^2\log n+|\Psibp|\log n)$. 
%Since $|\Psibp|=\Omega(n^2)$ in the worst case as discussed in Section~\ref{sec:pre}, we conclude as follows. 

\begin{theorem}\label{theo:algopsibp}
There is an algorithm that can compute $\Psibp$ in $O((n^2+|\Psibp|)\cdot \log n)$ time.     
\end{theorem}

\subsection{Upper bound of $\boldsymbol{|\Psibp|}$}
We prove that $|\Psibp|$ is bounded by $O(n^5)$ in the following theorem, by a simple brute-force argument. Similar ideas have been used elsewhere in the previous work~\cite{ref:BaeTh13,ref:MitchellSh96,ref:WangOn18}.

\begin{theorem}\label{theo:sizepsibp}
The combinatorial complexity of $\Psibp$ is $O(n^5)$. 
\end{theorem}
\begin{proof}
According to our previous discussion, when $s$ moves on a segment $e_s$ of $\Psisptb$, there are two type of events: (1) The event is caused by a combinatorial change of $\spm(s)$ on $\calB$; (2) the event is caused by a combinatorial change of $\spm(s)$ in the interior of $\calP$. 

The first type of events is actually included in $\Psibb$, whose complexity is $O(n^{4+\epsilon})$ by Corollary~\ref{coro:psibb}. Below, we show that the number of the second type events is $O(n^5)$.

For every four (not necessarily distinct) obstacle vertices $v_i$, $1\leq i\leq 4$, we do the following. Recall that $\spmb(v_i)$ has $O(n)$ segments on $\calB$. We overlay $\spmb(v_i)$, for all $1\leq i\leq 4$, resulting in $O(n)$ segments on $\calB$. For each such segment $e$, for each $1\leq i\leq 4$, $e$ is contained in a single cell of $\spm(v_i)$ and let $u_i$ be the root of the cell. Consider the equation $d_{u_1,v_1}(s,t)=d_{u_2,v_2}(s,t)=d_{u_3,v_3}(s,t)=d_{u_4,v_4}(s,t)$, for $s\in e$ and $t\in\bbR^2$. Observe that each second type event with $s\in e$ corresponds to a pair $(s,t)$ with $t$ in the interior of $\calP$ such that there are four shortest \st\ paths (see Fig.~\ref{fig:bpevent10}(b)). If the anchors of $t$ in the four shortest paths are $v_i$, $1\leq i\leq 4$, respectively, then we say that $s$ is {\em defined by} $v_i$, $1\leq i\leq 4$, and $(s,t)$ must be a solution to the above equation. Since the equation has $O(1)$ solutions, $e$ contains $O(1)$ second type events for $s$ defined by $v_i$, $1\leq i\leq 4$. As the above overlay has $O(n)$ segments, there are $O(n)$ second type events defined by $v_i$, $1\leq i\leq 4$. Enumerating all combinations of four obstacle vertices gives the $O(n^5)$ upper bound for the total number of second type events. 
\end{proof}

Combining Theorems \ref{theo:algopsibp} and \ref{theo:sizepsibp} leads to the following result. 
\begin{corollary}
The decomposition $\Psibp$ can be computed in $O(n^5\log n)$ time. 
\end{corollary}

\section{The decomposition $\boldsymbol{\Psipp}$}
\label{sec:psipp}

This section considers the most general decomposition $\Psipp$, which is $\Psi$. 
%We discuss the general decomposition $\Psi$, which is $\Psipp$. 
While we do not have a better upper bound than $O(n^{10})$ in~\cite{ref:ChiangTw99} on its combinatorial complexity, we give an algorithm of $O(n^{7.73}+|\Psi|\log n)$ time for constructing $\Psi$. Note that an $O(n^{10}\log n)$-time algorithm is presented in \cite{ref:ChiangTw99}, but the runtime does not depend on $|\Psi|$. Also, the algorithm for computing $\Psibp$ in Section~\ref{sec:psipb} is a simplified version of the algorithm in this section. 

For each obstacle vertex $u$, we define $f_u(s,t)$ in the same way as in Definition~\ref{def:fun} but with respect to $s\in \vis(u)$ and 
$t\in \calP$. As such, $f_u(s,t)$ is a 4-variate function defining an algebraic surface patch in $\bbR^5$. Similarly, we can decompose $f_u(s,t)$ into $O(n^2)$ {\em elementary functions} as follows. 
Let $\vis_{\triangle}(u)$ be a triangulation of $\vis(u)$. For each triangle $\triangle\in \vis_{\triangle}(u)$, for each cell $\sigma\in \spm(u)$, $f_u(s,t)$ on $s\in \triangle$ and $t\in \sigma$ is a constant-sized function since $f_u(s,t)=d_{u,v}(s,t)$, where $v$ is the root of the cell $\sigma$,  
and we call it an {\em elementary function} of $f_u(s,t)$. As $\vis_{\triangle}(u)$ has $O(n)$ triangles and $\spm(u)$ has $O(n)$ cells, $f_u(s,t)$ has $O(n^2)$ elementary functions. We refer to the original $f_u(s,t)$ as a {\em compound function}. 

For simplicity, we call the first two coordinates of $\bbR^5$ the {\em $s$-plane}, where the domain of $s$ lies; similarly, the third and fourth coordinates form the {\em $t$-plane}. The $s$-plane and $t$-plane together form the {\em $st$-space}. For any point $p$ in $\bbR^5$, we use $\eta_s(p)$ (resp., $\eta_t(p)$) to denote the projection of $p$ onto the $s$-plane (resp., $t$-plane). By slightly abusing the notation, for a set $\gamma$ of points in $\bbR^5$ (e.g., $\gamma$ is a curve), let $\eta_s(\gamma)=\{\eta_s(p)\ |\ p\in \gamma\}$; $\eta_t(\gamma)$ is defined similarly.

Let $F$ denote the set of functions for all obstacle vertices $u$. 
As such, $F$ has $O(n)$ compound functions with a total of $O(n^3)$ elementary functions.  
Let $\calL(F)$ denote the lower envelope of $F$ in $\bbR^5$. Denote by $\calA(F)$ the arrangement of the surface patches in $\bbR^5$ defined by the functions of $F$. 

%For any pair of points $s,t\in \calP$, let $z(s,t)$ be the (unique) point of $\calL(F)$ that intersect the vertical line (i.e., the line parallel to the $5$-th coordinate) through the point in $\bbR^5$ with the first four coordinates $(x(s),y(s),x(t),y(t))$. Observe that the $5$-th coordinate of $z(s,t)$ is equal to $d(s,t)$. 

\paragraph{Main idea.}
Roughly speaking, each interior vertex of $\calL(F)$ is the common intersection of five elementary functions of $F$ while each interior edge of $\calL(F)$ is the common intersection of four elementary functions. Let $e$ be an interior edge of $\calL(F)$. By our definition of the functions of $F$, $\eta_s(e)$ is a curve in $\calP$. Further, when $s$ moves across $\eta_s(e)$ in $\calP$, $\spm(s)$ changes topologically. Indeed, $\eta_s(e)$ belongs to a so-called {\em topological curve} defined in \cite{ref:ChiangTw99}. The reverse is also true, i.e., when $\spm(s)$ changes topologically during the moving of $s$, $s$ must cross a topological curve in $\calP$~\cite{ref:ChiangTw99}, and thus cross $\eta_s(e)$ for an interior edge $e$ of $\calL(F)$. 
%by the definition of functions of $F$. 
As such, $\Psi$ is the decomposition of $\calP$ by the curves of $E$, where $E$ is the set of projections of all interior edges of $\calL(F)$ on the $s$-plane. Once $E$ is available, $\Psi$ can be computed in $O((|E|+|\Psi|)\log n)$ time, e.g, by a plane sweeping algorithm. 

\paragraph{Remark.} Chiang and Mitchell's algorithm~\cite{ref:ChiangTw99} computes $\Psi$ by first computing a set of $O(n^6)$ topological curves and then computing the decomposition of $\calP$ by these curves. Their algorithm complexity does not depend on $|\Psi|$ because each topological curve may contain portions that are not projections of edges of $\calL(F)$ and thus those portions are redundant for $\Psi$. Hence,  their algorithm computes a refinement of $\Psi$. In contrast, our algorithm will identify the non-redundant portions of these topological curves and use them to decompose $\calP$ to obtain $\Psi$ in an output-sensitive way. 
\medskip

\paragraph{Combinatorial complexity of $\boldsymbol{\calL(F)}$.}
As $F$ has $O(n^3)$ $4$-variate elementary functions, applying the results of~\cite{ref:SharirAl94,ref:SharirDa95} directly would give an $O(n^{12+\epsilon})$ upper bound on the complexity of $\calL(F)$. In the following, we first show in Lemma~\ref{lem:psilesize} that $|\calL(F)|=O(n^7)$ by a simple brute-force method. Then, we present an $O(n^{7.73})$-time algorithm to compute all vertices and edges of $\calL(F)$, and thus the set $E$ (whose size is $O(n^7)$) can be obtained as well. Consequently, as discussed above, $\Psi$ can be constructed in additional $O((n^7+|\Psi|)\log n)$ time. 

\begin{lemma}\label{lem:psilesize}
The combinatorial complexity of $\calL(F)$ is bounded by $O(n^7)$. 
\end{lemma}
\begin{proof}
We prove it by a simple brute-force argument, somewhat similar to the proof of Theorem~\ref{theo:sizepsibp}.
%was used before~\cite{ref:BaeTh13,ref:MitchellSh96}.
It suffices to bound the number of vertices of $\calL(F)$. 

Since $\calL(F)$ is in $\bbR^5$, every interior vertex is bounded by five elementary functions. As such, we can bound the number of interior vertices as follows. For every five (not necessarily distinct) obstacle vertices $v_i$, $1\leq i\leq 5$, we overlay $\spm(v_i)$, for all $1\leq i\leq 5$. The overlay has $O(n^2)$ cells since each shortest path map is of size $O(n)$. For each cell $\sigma$ of the overlay, for each $1\leq i\leq 5$, $\sigma$ is contained in a single cell of $\spm(v_i)$ and let $u_i$ denote the root of the cell. Consider the equation $d_{u_1,v_1}(s,t)=d_{u_2,v_2}(s,t)=d_{u_3,v_3}(s,t)=d_{u_4,v_4}(s,t)=d_{u_5,v_5}(s,t)$, for $s\in \sigma$ and $t\in \bbR^2$. 
Consider an interior vertex $p$ of $\calL(F)$. Let $s=\eta_s(p)$ and $t=\eta_t(p)$. 
Since $p$ is the common intersection of five elementary functions, there are five shortest \st\ paths. Suppose the anchors of $t$ in the five shortest paths are $v_i$, $1\leq i\leq 5$, respectively, and $s$ is in the cell $\sigma$; we say that $p$ is {\em defined by} $v_i$, $1\leq i\leq 5$, and the cell $\sigma$. Then, $(s,t)$ must satisfies the above equation. As such, since the equation has $O(1)$ solutions, the cell $\sigma$, along with $v_i$, $1\leq i\leq 5$, can define $O(1)$ interior vertices of $\calL(F)$. As the overlay has $O(n^2)$ cells, the five obstacle vertices $v_i$, $1\leq i\leq 5$, define $O(n^2)$ interior vertices for $\calL(F)$. Enumerating all $O(n^5)$ possible combinations of five obstacle vertices gives the $O(n^7)$ upper bound on the number of interior vertices of $\calL(F)$. 

Next we discuss the boundary vertices of $\calL(F)$. Each boundary vertex $p$ has $s$ or/and $t$ on $\calB$, where $s=\eta_s(p)$ and $t=\eta_t(p)$.
%are the projections of $p$ onto the $s$-plane and $t$-plane, respectively. 
%Boundary vertices correspond to $(s,t)$ with at least one of $s$ and $t$ on $\calB$. 

\begin{itemize}
\item 
If both $s$ and $t$ are on $\calB$, then observe that $(s,t)$ corresponds to an event point of the decomposition $\Psibb$. Since $\Psibb=O(n^{4+\epsilon})$ by Corollary~\ref{coro:psibb}, the number of such boundary vertices $p$ is $O(n^{4+\epsilon})$. 
\item 
If $s\in \calB$ and $t\in \calP\setminus\calB$, then observe that $(s,t)$ corresponds to an event point of the decomposition $\Psibp$. Since $\Psibp=O(n^{5})$ by Theorem~\ref{theo:sizepsibp}, the number of such boundary vertices $p$ is $O(n^{5})$. 
\item 
If $s\in \calP\setminus\calB$ and $t\in \calB$, then as will be discussed in Section~\ref{sec:psipb}, $(s,t)$ corresponds to a vertex in the lower envelope of the functions defined for the decomposition $\Psipb$, which has $O(n^5)$ vertices by  Lemma~\ref{lem:psipblesize}. Hence, the number of such boundary vertices $p$ is $O(n^5)$. 
\end{itemize}

The lemma thus follows. 
\end{proof}

%\subsection{Computing the vertices and edges of $\calL(F)$ and $\calE$}
\subsection{Computing the vertices and edges of $\boldsymbol{\calL(F)}$}
In the following, we present an algorithm to compute the vertices and edges of $\calL(F)$. It suffices to compute the edges of $\calL(F)$ as the endpoints of edges are vertices of $\calL(F)$.  
%and the projections of these edges on the $s$-plane is $\calE$. 

Since $F$ has $O(n^3)$ $4$-variate elementary  functions, applying the algorithm of \cite{ref:AgarwalCo97} directly can compute all edges of $\calL(F)$ in $O(n^{12+\epsilon})$ expected time by a randomized algorithm. In what follows, we give a more efficient (deterministic) algorithm of $O(n^{7.73})$ time.

Below, we focus on computing the interior edges of $\calL(F)$.  
For the boundary edges of $\calL(F)$, they are actually edges in the lower envelope of the functions for the decomposition $\Psipb$ in Section~\ref{sec:psipb}, and the algorithm there, which essentially solves the problem in one dimension lower and computes these boundary edges in $O(n^5\log^2 n)$ time (the number of such edges is $O(n^5)$), is a simplified version of the algorithm given below for computing the interior edges of $\calL(F)$. 

Consider a set $V$ of four (not necessarily distinct) obstacle vertices $u_i$, $1\leq i\leq 4$. 
%The four vertices may not be distinct. 
We name other vertices as $u_5,u_6,\ldots,u_n$ (if $u_1,u_2,u_3,u_4$ are not distinct, then we can make $n$ larger so that the number of other obstacle vertices is always $n-4$). 
Below, we compute the edges of $\calL(F)$ that are defined by the common intersections of functions of $f_{u_i}(s,t)$, $1\leq i\leq 4$. 

We compute the overlay of $\spm(u_i)$, $1\leq i\leq 4$, denoted by $\calO_V$. 
%We further compute a vertical decomposition of every cell of $\calO_V$ (e.g., by introducing vertical segments from vertices as well as from tangents of curves) so that each cell is of constant size. The size of the new $\calO_V$ is still $O(n^2)$ size. 
For each cell $\sigma\in \calO_V$, for each $1\leq i\leq 4$, $\sigma$ is contained in a single cell of $\spm(u_i)$ and let $v_i$ denote the root of the cell. Solve the following equation $d_{u_1,v_1}(s,t)=d_{u_2,v_2}(s,t)=d_{u_3,v_3}(s,t)=d_{u_4,v_4}(s,t)$ for $s\in \bbR^2$ and $t\in \bbR^2$ gives a curve $\gamma_{\sigma}$ in $\bbR^5$, 
%(if the solution of the equation contains a plane, then we ignore it), 
which may contain edges of $\calL(F)$ that are common intersections of the functions $f_{u_i}$, $1\leq i\leq 4$. Our main goal is to identify these edges. 

Recall that $\eta_s(\gamma_{\sigma})$ (resp, $\eta_t(\gamma_{\sigma})$) denote the projection of $\gamma_{\sigma}$ onto the $s$-plane (resp., the $t$-plane), which is a curve.

\paragraph{Remark.}
In Chiang and Mitchell's algorithm for computing $\Psi$~\cite{ref:ChiangTw99}, $\eta_s(\gamma_{\sigma})$ is called a {\em topological curve}. Since $\calO_V$ has $O(n^2)$ cells, we can obtain $O(n^2)$ topological curves defined by $V$. Enumerating all subsets $V$ of four obstacle vertices produces $O(n^6)$ topological curves. Then, a decomposition of $\calP$ by these topological curves (along with other topological curves for the case where at least one of $s$ and $t$ is on $\calB$) is computed and returned as $\Psi$. They proved that the number of intersections of those $O(n^6)$ topological curves is bounded by $O(n^{10})$~\cite{ref:ChiangTw99}. 
\medskip

We say that a point $p\in \gamma_{\sigma}$ is {\em valid} if $(\eta_s(p),\eta_t(p))$ satisfies the following conditions: (1) $\eta_s(p)$ is visible to $u_i$, for each $1\leq i\leq 4$; (2) $\eta_t(p)$ is in the cell $\sigma$. A sub-curve of $\gamma_{\sigma}$ is {\em valid} if all its points are valid. We first have the following observation. 

\begin{observation}\label{obser:sufficient}
Each valid sub-curve of $\gamma_{\sigma}$ consists of a sequence of edges of $\calA(F)$, i.e., every valid point of $\gamma_{\sigma}$ is on an edge of $\calA(F)$. 
\end{observation}
\begin{proof}
Consider a valid point $p$ of $\gamma_{\sigma}$. We argue that $p$ must be on an edge of $\calA(F)$. Indeed, by definition, $p$ is in the common intersection of the elementary functions $d_{u_i,v_i}(s,t)$, $1\leq i\leq 4$. Since the equation $d_{u_1,v_1}(s,t)=d_{u_2,v_2}(s,t)=d_{u_3,v_3}(s,t)=d_{u_4,v_4}(s,t)$ gives a curve $\gamma_{\sigma}$ in $\bbR^5$, each point in the common intersection of the above elementary functions must belong to an edge of $\calA(F)$. As such, $p$ is on an edge of $\calA(F)$.
\end{proof}

On the other hand, the following observation shows that each edge of $\calL(F)$ is contained in a valid sub-curve.

\begin{observation}\label{obser:necessary}
For any point $p$ in an interior edge of $\calL(F)$, there must be a set $V$ of four obstacle vertices $u_i$, $1\leq i\leq 4$, such that $p$ is a valid point on the curve $\gamma_{\sigma}$ defined by a cell $\sigma\in \calO_V$.   
\end{observation}
\begin{proof}
Recall that each interior edge of $\calL(F)$ belongs to the  common intersection of four elementary functions of $F$. Let $e$ be the interior edge of $\calL(F)$ that contains $p$. 
Since $e$ belongs to the common intersection of four elementary functions, say,  
$d_{u_i,v_i}(s,t)$, for some obstacle vertices $u_i$ and $v_i$, $1\leq i\leq 4$, and $e\in \calL(F)$, we have  $d_{u_1,v_1}(s,t)=d_{u_2,v_2}(s,t)=d_{u_3,v_3}(s,t)=d_{u_4,v_4}(s,t)$ and $\overline{su_i}\cup\pi(u_i,v_i)\cup \overline{v_it}$ must be a shortest \st\ path for each $1\leq i\leq 4$. As such, $s$ must be visible to $u_i$ for all $1\leq i\leq 4$, and there must be a cell $\sigma$ in the overlay of $\spm(u_i)$, $1\leq i\leq 4$, such that (1) the root of the cell of $\spm(u_i)$ containing $\sigma$ is $v_i$, for each $1\leq i\leq 4$; (2) $t\in \sigma$. By definition, $p$ must be a valid point on $\gamma_{\sigma}$. 
\end{proof}

\paragraph{Algorithm overview.}
%In light of Observation~\ref{obser:necessary}, 
We first find all valid sub-curves of $\gamma_{\sigma}$. By Observation~\ref{obser:sufficient}, each valid sub-curve comprises a sequence of edges of $\calA(F)$. We next find all vertices of $\calA(F)$ on each valid sub-curve, after which each portion of the curve between two vertices is an edge of $\calA(F)$. In this way, all edges contained in each valid curve are computed. We finally determine for each edge whether it is an edge on $\calL(F)$. If we do this for all cells of $\calO_V$ and for all subsets $V$ of four obstacle vertices, then all edges of $\calL(F)$ are computed by Observation~\ref{obser:necessary}. As such, our algorithm proceeds with three main procedures. First, find all valid sub-curves of $\gamma_{\sigma}$. Second, compute all vertices on all valid sub-curves. Third, for each edge in these valid sub-curves, determine whether it is on $\calL(F)$. 
\medskip

%Then, we process each valid sub-curve $\xi$ as follows. For each obstacle vertex $u_j$ with $j\geq 5$, define $\xi^+_j=\{(s,t) \ |\ f_{u_j}(s,t)> f_{u_1}(s,t)\}$ (we assume $f_{u_j}(s,t)=\infty$ if $s$ is not visible to $u_j$). Then, the common intersection $\bigcup_{j\geq 5}\xi^+_j$ is the portion of $\xi$ over which $\calL(F)$ is attained by the elementary functions $d_{u_i,v_i}(s,t)$, $1\leq i\leq 4$. Doing this for all cells $\sigma$ of $\calO_V$ will determine the edges of $\calL(F)$ defined by the functions $f_{u_i}(s,t)$, $1\leq i\leq 4$. Doing this for all subsets $V$ of four obstacle vertices will compute all interior edges of $\calL(F)$. We next give the details of the algorithm. We first present a preliminary algorithm of $O(n^8\log n)$ time and then improve it to $O(n^{7.73})$ time. 

\paragraph{The first procedure: Finding the valid sub-curves.}
We wish to partition $\gamma_{\sigma}$ into maximal sub-curves whose projections on the $t$-plane are inside $\sigma$; let $\Gamma^*_{\sigma}$ denote the set of such sub-curves. 
%we call these sub-curves {\em candidate-valid sub-curves}. 
We can compute $\Gamma^*_{\sigma}$ as follows. We compute a vertical decomposition of $\sigma$ (i.e., by introducing vertical segments from vertices as well as from locally $x$-extreme points of curves) and we also compute a vertical decomposition for $\bbR^2\setminus \sigma$, i.e., the region of the plane outside $\sigma$; we use $\triangle(\sigma)$ to denote the resulting decomposition of the plane. Since $\sigma$ is of size $O(n)$, $\triangle(\sigma)$ can be easily computed in $O(n\log n)$ time, e.g., by a plane sweeping algorithm. Then, by traversing on $\eta_t(\gamma_{\sigma})$ using $\triangle(\sigma)$, we can compute $\Gamma^*_{\sigma}$ in $O(n)$ time (indeed, since $\eta_t(\gamma_{\sigma})$ is of constant size, it intersects each edge of $\triangle(\sigma)$ a constant number of times). Hence, $|\Gamma^*_{\sigma}|=O(n)$. 

Next, for each $1\leq i\leq 4$, we find the sub-curves of $\gamma_{\sigma}$ whose projections on the $s$-plane are inside $\vis(u_i)$; let $\Gamma^i_{\sigma}$ denote the set of these sub-curves. We can compute $\Gamma^i_{\sigma}$ in $O(n\log n)$ time, e.g., by first triangulating $\vis(u_i)$ as well as $\bbR^2\setminus \vis(u_i)$ and then traversing on $\eta_s(\gamma_{\sigma})$ using the triangulation. As $\eta_s(\gamma_{\sigma})$ is of constant size, it intersects each edge of $\vis(u_i)$ at most $O(1)$ times. Hence, $|\Gamma^i_{\sigma}|=O(n)$. 

Now we compute the common intersection $\Gamma_{\sigma}$ of the sub-curves of $\Gamma^*_{\sigma}$ and $\Gamma^i_{\sigma}$ for all $1\leq i\leq 4$, i.e., a point $p$ is in a sub-curve of $\Gamma_{\sigma}$ if $p$ is in a curve of $\Gamma^*_{\sigma}$ and also in a curve of $\Gamma^i_{\sigma}$ for every $1\leq i\leq 4$. Since the total number of sub-curves of $\Gamma^*_{\sigma}$ and $\Gamma^i_{\sigma}$, $1\leq i\leq 4$, is $O(n)$, $\Gamma_{\sigma}$ has $O(n)$ sub-curves of $\gamma_{\sigma}$. By definition, $\Gamma_{\sigma}$ is exactly the set of valid sub-curves of $\gamma_{\sigma}$. 
We can compute $\Gamma_{\sigma}$ in $O(n\log n)$ time as follows. 

Let $\Gamma$ denote the union of $\Gamma^*_{\sigma}$ and $\Gamma^i_{\sigma}$ for all $1\leq i\leq 4$. 
Our algorithm is based on the observation that a point $p$ is in a curve of $\Gamma_{\sigma}$ if and only if $p$ is covered by exactly $5$ sub-curves of $\Gamma$, i.e., by a sub-curve of $\Gamma^*_{\sigma}$ and a sub-curve of $\Gamma^i_{\sigma}$ for every $1\leq i\leq 4$. 
We first sort the endpoints of the sub-curves of $\Gamma$ on $\gamma_{\sigma}$. Then, we sweep a point $p$ on $\gamma_{\sigma}$ from one end to the other. During the sweeping, we maintain a counter $c$ to record the number of sub-curves that currently cover $p$. An event happens if $p$ meets an endpoint of a sub-curve of $\Gamma$. If $p$ is entering the sub-curve, then we increment $c$ by one. If $c=5$, this means that we are entering a common intersection now; let $q=p$, i.e., $q$ is used to record the first endpoint of a common intersection. If $p$ is leaving the sub-curve, then we do the following. If $c=5$, then we are leaving a common intersection whose right endpoint is $p$, and thus we report the portion of $\gamma_{\sigma}$ from $q$ to $p$ and we also decrement $c$ by one. 
If $c<5$, we just decrement $c$ by one. 

As such, we can compute the set $\Gamma_{\sigma}$ of $O(n)$ valid sub-curves of $\gamma_{\sigma}$ in $O(n\log n)$ time. Doing this for all $O(n^2)$ cells $\sigma$ of the overlay $\calO_V$ computes in $O(n^3\log n)$ time a total of $O(n^3)$ valid sub-curves for $O(n^2)$ curves $\gamma_{\sigma}$. Let $\Gamma(V)$ denote the set of all of these $O(n^3)$ valid curves. 

This finishes the first procedure, which takes $O(n^3\log n)$ time in total. 

\paragraph{The second procedure: Computing the vertices on valid curves.}
We now compute the vertices of $\calA(F)$ on all valid curves of $\Gamma(V)$. Our approach relies on the fact that a vertex on a curve $\gamma_{\sigma}$ must be the intersection between $\gamma_{\sigma}$ and another elementary function $d_{u_j,v_j}(s,t)$ for some $j>4$. 

For each $u_j$, $j>4$, we do the following. For notational convenience, let $j=5$. 
We compute the overlay $\calO_j$ of the five shortest path maps $\spm(u_i)$, $1\leq i\leq 5$, which takes $O(n^2\log n)$ time. For each cell $\sigma'$ of $\calO_j$, for each $1\leq i\leq 5$, $\sigma'$ is contained in a single cell of $\spm(u_i)$ and let $v_i$ denote the root of the cell. Also, $\sigma'$ must be contained in a single cell of $\calO_V$, denoted by $\sigma$. 
We solve the equation $d_{u_1,v_1}(s,t)=d_{u_2,v_2}(s,t)=d_{u_3,v_3}(s,t)=d_{u_4,v_4}(s,t)=d_{u_5,v_5}(s,t)$ for $s\in \bbR^2$ and $t\in \bbR^2$ to obtain $O(1)$ solutions $(s,t)$. Each such solution $(s,t)$ is processed as follows. First, we do a point location to find the cell of $\calO_j$ that contains $t$; if the cell is not $\sigma'$, then we ignore this solution. Second, we check whether $s\in \vis(u_i)$ (which can be done in $O(\log n)$ time by constructing a point location data structure on $\vis(u_i)$), for every $1\leq i\leq 5$. If $s$ is not in $\vis(u_i)$ for any $1\leq i\leq 5$, then we ignore this solution $(s,t)$. If $t\in \sigma'\subseteq \sigma$ and $s\in \vis(u_i)$ for every $1\leq i\leq 5$, then the point $p=(s,t,d_{u_1,v_1}(s,t))\in\bbR^5$ must be on the curve $\gamma_{\sigma}$. We further check whether $p$ is on a valid curve of $\Gamma_{\sigma}$, which can be done in $O(n)$ time by binary search after curves of $\Gamma_{\sigma}$ are sorted along $\gamma_{\sigma}$; if yes, then $p$ is a vertex and we insert it to the valid curve that contains it. 

As $\calO_j$ has $O(n^2)$ cells, the above computes in $O(n^2\log n)$ time $O(n^2)$ vertices for the curves of $\Gamma(V)$. Doing this for all $u_j$ with $j>4$ computing all $O(n^3)$ vertices for the curves of $\Gamma(V)$. 
For each curve of $\Gamma(V)$, by sorting all its vertices, we obtain the edges of $\calA(F)$ contained in the curve. As $|\Gamma(V)|=O(n^3)$, this step takes $O(n^3\log n)$ time in total. Let $E(V)$ denote the set of edges thus computed. 

This finishes the second procedure, which takes $O(n^3\log n)$ time in total. 

\paragraph{The third procedure: Identifying the edges of $\boldsymbol{\calL(F)}$.}
We now find the edges of $E(V)$ that are on $\calL(F)$. Let $e$ be an edge of $E(V)$. To determine whether $e\in \calL(F)$, we utilize the observation that $e\in \calL(F)$ if and only if $p\in \calL(F)$ for any interior point $p$ of $e$. We pick an interior point $p$ on $e$. Let $s=\eta_s(p)$, $t=\eta_t(p)$, and $d(p)$ be the $5$-th coordinate value of $p$. Observe that $p\in \calL(F)$ if and only if $d(p)=d(s,t)$. As such, the problem boils down to computing the geodesic distance $d(s,t)$. 
We compute $d(s,t)$ by a two-point shortest path query.

Suppose there is a two-point shortest path query data structure of $T(n)$ preprocessing time and $Q(n)$ query time.
Then, finding all edges of $E(V)$ that are on $\calL(F)$ can be done in $O(n^3\cdot Q(n))$ time as $|E(V)|=O(n^3)$. 

Doing this for all subsets $V$ of four obstacle vertices can compute all edges of $\calL(F)$ in $O(n^7\log n+T(n)+n^7\cdot Q(n))$ time. The following theorem summarizes our result. 

\begin{theorem}\label{theo:constructpsi}
Suppose there is a two-point shortest path query data structure of $T(n)$ preprocessing time and $Q(n)$ query time. Then, one can compute all edges and vertices of $\calL(F)$ in $O(n^7\log n+T(n)+n^7\cdot Q(n))$ time. The decomposition $\Psi$ of size $O(n^7+|\Psi|)$ can be computed in additional $O((n^7+|\Psi|)\log n)$ time. 
\end{theorem}

\paragraph{Remark.} We will discuss in Section~\ref{sec:query} that $\Psi$ can be used to construct a two-point shortest path query data structure. Theorem~\ref{theo:constructpsi} is interesting in the sense that it shows that a two-point shortest path query data structure can be used to construct $\Psi$, implying that the two problems can be somehow reduced to each other. 
\medskip

As discussed in~\cite{ref:BaeTh13}, using the preprocessing-query trade-off developed by Chiang and Mitchell~\cite{ref:ChiangTw99}, we can have $T(n)=O(n^{7.73})$ and $Q(n)=O(n^{0.73})$, or $T(n)=O(n^{5})$ and $Q(n)=O(h+\log n)$. Consequently, we have the following. 

\begin{corollary}\label{coro:psipp}
One can compute all edges and vertices of $\calL(F)$ in $O(\min\{n^{7.73},n^7\cdot (h+\log n)\})$ time. The decomposition $\Psi$ of size $O(n^7+|\Psi|)$ can be computed in additional $O((n^7+|\Psi|)\log n)$ time. 
\end{corollary}

De Berg, Miltzow, and Staals~\cite{ref:deBergTo24} also provided a preprocessing-query trade-off, with a randomized preprocessing time. Using their result, one could obtain a randomized time with a slightly smaller term than $O(n^{7.73})$.

%\paragraph{Remark.} For comparison, Chiang and Mitchell's approach~\cite{ref:ChiangTw99} first computes a collection of topological curves and then computes the decomposition of $\calP$ by these topological curves, which was proved to be of size $O(n^{10})$. Their topological curve is something similar to our curve $\xi_s$ and therefore may be a compound set of the projection of the edges of $\calL(F)$ on the $s$-plane. As such, their algorithm is not out-sensitive. In particular, a smaller upper bound of $|\Psi|$ would not immediately lead to an improvement of their algorithm. 

\section{The decomposition $\boldsymbol{\Psipb}$}
\label{sec:psipb}

In this section, we discuss the decomposition $\Psipb$. Recall that $\Psipb$ is the decomposition of $\calP$ into cells such that $\spmb(s)$ is combinatorially the same for all points $s$ in the same cell.
While we do not have a better upper bound than the bound $O(n^{10})$ in~\cite{ref:ChiangTw99}, we give an algorithm of $O((n^{5}+|\Psipb|)\log n)$ time for constructing $\Psipb$. Note that an $O(n^{10}\log n)$-time algorithm is presented in \cite{ref:ChiangTw99}, but its runtime does not depend on $|\Psipb|$.

The algorithm in this section is a simplified version of that in Section~\ref{sec:psipp}. We follow the structure and notation from there.

For each obstacle vertex $u$, we define $f_u(s,t)$ in the same way as in Definition~\ref{def:fun} but with respect to $s\in \vis(u)$ and 
$t\in \calB$. As such, $f_u(s,t)$ is a 3-variate function defining an algebraic surface patch in $\bbR^4$. Similarly, we can decompose $f_u(s,t)$ into $O(n^2)$ {\em elementary functions} as follows. 
Let $\vis_{\triangle}(u)$ be a triangulation of $\vis(u)$. For each triangle $\triangle\in \vis_{\triangle}(u)$, for each segment $\sigma\in \spmb(u)$, $f_u(s,t)$ on $s\in \triangle$ and $t\in \sigma$ is a constant-sized function since $f_u(s,t)=d_{u,v}(s,t)$, where $v$ is the root of the cell of $\spm(u)$ containing $\sigma$,  
and we call it an {\em elementary function} of $f_u(s,t)$. As $\vis_{\triangle}(u)$ has $O(n)$ triangles and $\spmb(u)$ has $O(n)$ segments, $f_u(s,t)$ has $O(n^2)$ elementary functions. 
We refer to the original function $f_u(s,t)$ as a {\em compound function}. 

For simplicity, we call the first two coordinates of $\bbR^4$ the {\em $s$-plane}, where the domain of $s$ lies; we refer to the third coordinate as the {\em $t$-axis} in the same way as defined in Section~\ref{sec:psibb}. 
%The $s$-plane and $t$-plane together form the {\em $st$-space}. 
For any point $p$ in $\bbR^4$, we use $\eta_s(p)$ (resp., $\eta_t(p)$) to denote the projection of $p$ onto the $s$-plane (resp., $t$-axis). By slightly abusing the notation, for a set $\gamma$ of points in $\bbR^4$ (e.g., $\gamma$ is a curve), let $\eta_s(\gamma)=\{\eta_s(p)\ |\ p\in \gamma\}$; $\eta_t(\gamma)$ is defined similarly.

Let $F$ denote the set of functions $f_u$ for all obstacle vertices $u$. 
As such, $F$ has $O(n)$ compound functions with a total of $O(n^3)$ elementary functions.  
Let $\calL(F)$ denote the lower envelope of $F$ in $\bbR^4$. Denote by $\calA(F)$ the arrangement of the surface patches in $\bbR^4$ defined by the functions of $F$. 

\paragraph{Main idea.}
Roughly speaking, each interior vertex of $\calL(F)$ is the common intersection of four elementary functions of $F$ while each interior edge of $\calL(F)$ is the common intersection of three elementary functions. Let $e$ be an interior edge of $\calL(F)$. By our definition of the functions of $F$, $\eta_s(e)$ is a curve in $\calP$. Further, when $s$ moves across $\eta_s(e)$ in $\calP$, $\spmb(s)$ changes topologically. Indeed, $\eta_s(e)$ belongs to a so-called {\em topological curve} defined in \cite{ref:ChiangTw99}. The reverse is also true, i.e., when $\spmb(s)$ changes topologically during the moving of $s$, $s$ must cross a topological curve in $\calP$~\cite{ref:ChiangTw99}, and thus cross $\eta_s(e)$ for an interior edge $e$ of $\calL(F)$. As such, $\Psipb$ is the decomposition of $\calP$ by the curves of $E$, where $E$ is the set of projections of all interior edges of $\calL(F)$ on the $s$-plane. Once $E$ is available, $\Psi$ can be computed in $O((|E|+|\Psipb|)\log n)$ time, e.g, by a plane sweeping algorithm. 

\paragraph{Combinatorial complexity of $\boldsymbol{\calL(F)}$.}
As $F$ has $O(n^3)$ $3$-variate elementary functions, applying the results of~\cite{ref:SharirAl94,ref:SharirDa95} directly would give an $O(n^{9+\epsilon})$ upper bound on the complexity of $\calL(F)$. In the following, we first show in Lemma~\ref{lem:psipblesize} that $|\calL(F)|=O(n^5)$ by a simple brute-force method. Then, we present an $O(n^{5}\log n)$-time algorithm to compute all vertices and edges of $\calL(F)$, and thus the set $E$ (whose size is $O(n^5)$) can be obtained as well. Consequently, as discussed above, $\Psipb$ can be constructed in additional $O((n^5+|\Psipb|)\log n)$ time. 

\begin{lemma}\label{lem:psipblesize}
The combinatorial complexity of $\calL(F)$ is bounded by $O(n^5)$. 
\end{lemma}
\begin{proof}
We prove it by a simple brute-force argument. Similar methods were used before~\cite{ref:BaeTh13,ref:ChiangTw99,ref:WangOn18}.
It suffices to bound the number of vertices of $\calL(F)$. 

Since $\calL(F)$ is in $\bbR^4$, every interior vertex is bounded by four elementary functions. Hence, we can bound the number of interior vertices as follows. For every four obstacle vertices $u_i$, $1\leq i\leq 4$, we consider the overlay of $\spmb(u_i)$, for all $1\leq i\leq 4$. The overlay has $O(n)$ segments on $\calB$ since each $\spmb(u_i)$ has $O(n)$ segments. For each segment $\sigma$ of the overlay, for each $1\leq i\leq 4$, $\sigma$ is contained in a single cell of $\spm(u_i)$ and let $v_i$ be the root of the cell. Consider the equation $d_{u_1,v_1}(s,t)=d_{u_2,v_2}(s,t)=d_{u_3,v_3}(s,t)=d_{u_4,v_4}(s,t)$, for $t\in \sigma$ and $s\in \bbR^2$. 
On the other hand, consider an interior vertex $p$ of $\calL(F)$. Let $s=\eta_s(p)$ and $t=\eta_t(p)$. 
Since $p$ is the common intersection of four elementary functions, there are four shortest \st\ paths. Suppose the anchors of $s$ in the four shortest paths are $u_i$, $1\leq i\leq 4$, respectively, and $t\in \sigma$; we say that $p$ is {\em defined by} $u_i$, $1\leq i\leq 4$, and the segment $\sigma$. Then, $(s,t)$ must satisfy the above equation. As such, since the equation has $O(1)$ solutions, the segment $\sigma$, along with $u_i$, $1\leq i\leq 4$, can define $O(1)$ interior vertices of $\calL(F)$. As the overlay has $O(n)$ segments, the four obstacle vertices $u_i$, $1\leq i\leq 4$, define $O(n)$ interior vertices for $\calL(F)$. Enumerating all $O(n^4)$ possible combinations of four obstacle vertices gives the $O(n^5)$ upper bound on the number of interior vertices of $\calL(F)$. 

For the boundary vertices of $\calL(F)$, each such vertex $p$ has both $s=\eta_s(p)$ and $t=\eta_t(p)$ on $\calB$. As such, $(s,t)$ corresponds to an event point of the decomposition $\Psibb$. Since $|\Psibb|=O(n^{4+\epsilon})$ by Corollary~\ref{coro:psibb}, the number of boundary vertices $p$ is $O(n^{4+\epsilon})$. 

In summary, the number of vertices and thus the combinatorial size of $\calL(F)$ is $O(n^5)$.
\end{proof}

\subsection{Computing the vertices and edges of $\boldsymbol{\calL(F)}$}
In the following, we present an algorithm to compute the vertices and edges of $\calL(F)$. It suffices to compute the edges of $\calL(F)$ as the endpoints of edges are vertices of $\calL(F)$.  
%and the projections of these edges on the $s$-plane is $\calE$. 

Since $F$ has $O(n^3)$ $4$-variate elementary  functions, applying the algorithm of \cite{ref:AgarwalCo97} directly can compute all edges of $\calL(F)$ in $O(n^{9+\epsilon})$ expected time by a randomized algorithm. In what follows, we give a more efficient (deterministic) algorithm of $O(n^{5}\log n)$ time.

Below, we focus on computing the interior edges of $\calL(F)$.  
As discussed in the proof of Lemma~\ref{lem:psipblesize}, points $p$ in the boundary edges of $\calL(F)$ have $\eta_s(p)$ and $\eta_t(p)$ on $\calB$. 
%Thus, boundary edges correspond to the edges of the decomposition $\Psibb$ (more specifically, the projection of each boundary edge on either the $s$-plane or the $t$-axis is an edge of $\Psibb$). Hence, boundary edges can be obtained from the edges of $\Psibb$; we omit the details. 
As such, boundary edges can be computed by a similar (but simpler) algorithm, which solves the problem in the space one dimension lower. 

Consider a set $V$ of three (not necessarily distinct) obstacle vertices $u_i$, $1\leq i\leq 3$. 
We name other vertices as $u_4,u_5,\ldots,u_n$. 
%(if $u_1,u_2,u_3$ are not distinct, then we can make $n$ larger so that the number of other obstacle vertices is always $n-3$). 
Below, we compute the edges of $\calL(F)$ that are defined by the common intersections of functions of $f_{u_i}(s,t)$, $1\leq i\leq 3$. 

We compute the overlay of $\spmb(u_i)$, $1\leq i\leq 3$, denoted by $\calO_V$. 
For each segment $\sigma\in \calO_V$, for each $1\leq i\leq 3$, $\sigma$ is contained in a single cell of $\spm(u_i)$ and let $v_i$ denote the root of the cell. The following equation $d_{u_1,v_1}(s,t)=d_{u_2,v_2}(s,t)=d_{u_3,v_3}(s,t)$ for $s\in \bbR^2$ and $t\in \bbR$ gives a $3$-dimensional curve $\gamma_{\sigma}$ in $\bbR^4$, 
%(if the solution of the equation contains a plane, then we ignore it), 
which may contain edges of $\calL(F)$ that are common intersections of the functions $f_{u_i}$, $1\leq i\leq 3$. Our goal is to identify these edges. 

Recall that $\eta_s(\gamma_{\sigma})$ (resp, $\eta_t(\gamma_{\sigma})$) denote the projection of $\gamma_{\sigma}$ onto the $s$-plane (resp., the $t$-axis).

We say that a point $p\in \gamma_{\sigma}$ is {\em valid} if $(\eta_s(p),\eta_t(p))$ satisfies the following conditions: (1) $\eta_s(p)$ is visible to $u_i$, for each $1\leq i\leq 3$; (2) $\eta_t(p)\in \sigma$. A sub-curve of $\gamma_{\sigma}$ is {\em valid} if all its points are valid. We first have the following observation. 

\begin{observation}\label{obser:sufficientpsipb}
Each valid sub-curve of $\gamma_{\sigma}$ consists of a sequence of edges of $\calA(F)$ i.e., every valid point of $\gamma_{\sigma}$ is on an edge of $\calA(F)$.  
\end{observation}
\begin{proof}
Consider a valid point $p$ of $\gamma_{\sigma}$. We argue that $p$ must be on an edge of $\calA(F)$. Indeed, by definition, $p$ is the common intersection of the elementary functions $d_{u_i,v_i}(s,t)$, $1\leq i\leq 3$. Since the equation $d_{u_1,v_1}(s,t)=d_{u_2,v_2}(s,t)=d_{u_3,v_3}(s,t)$ gives a curve $\gamma_{\sigma}$, each point in the common intersection of the above elementary functions must belong to an edge of $\calA(F)$. As such, $p$ is on an edge of $\calA(F)$.
\end{proof}

On the other hand, the following lemma shows that all edges of $\calL(F)$ are contained in valid sub-curves.

\begin{observation}\label{obser:necessarypsipb}
For any point $p$ in an interior edge of $\calL(F)$, there must be a set $V$ of three obstacle vertices $u_i$, $1\leq i\leq 3$, such that $p$ is a valid point on the curve $\gamma_{\sigma}$ defined by a segment $\sigma\in \calO_V$.   
\end{observation}
\begin{proof}
%Recall that that each interior edge of $\calL(F)$ belongs to the common intersection of four elementary functions of $F$. 
Let $e$ be the interior edge of $\calL(F)$ that contains $p$. 
Since $e$ belongs to the common intersection of three elementary functions, say,  
$d_{u_i,v_i}(s,t)$, for some obstacle vertices $u_i$ and $v_i$, $1\leq i\leq 3$, and $e\in \calL(F)$, we have $d_{u_1,v_1}(s,t)=d_{u_2,v_2}(s,t)=d_{u_3,v_3}(s,t)$ and $\overline{su_i}\cup\pi(u_i,v_i)\cup \overline{v_it}$ must be a shortest \st\ path for each $1\leq i\leq 3$. As such, $s$ must be visible to $u_i$ for all $1\leq i\leq 3$, and there must be a segment $\sigma$ in the overlay of $\spm(u_i)$, $1\leq i\leq 3$, such that (1) the root of the cell of $\spm(u_i)$ containing $\sigma$ is $v_i$, for each $1\leq i\leq 3$; (2) $t\in \sigma$. By definition, $p$ must be a valid point on $\gamma_{\sigma}$. 
\end{proof}

\paragraph{Algorithm overview.}
With Observations~\ref{obser:sufficientpsipb} and \ref{obser:necessarypsipb}, as in Section~\ref{sec:psipp}, 
%we first find all valid sub-curves of $\gamma_{\sigma}$. By Observation~\ref{obser:sufficientpsipb}, each valid sub-curve contains a sequence of edges of $\calA(F)$. We next find all vertices of $\calA(F)$ on each valid sub-curve, after which each portion of the curve between two vertices is an edge of $\calA(F)$. In this way, we find all edges contained in each valid curve. We finally determine for each edge whether it is an edge on $\calL(F)$. If we do this for all cells of $\calO_V$ and for all subsets $V$ of four obstacle vertices, then all edges of $\calL(F)$ are computed by Observation~\ref{obser:necessarypsipb}. As such, 
our algorithm has three procedures. First, find all valid sub-curves of $\gamma_{\sigma}$. Second, compute all vertices on valid sub-curves. Third, for each edge in these valid sub-curves, determine whether it is on $\calL(F)$. 
\medskip

%Then, we process each valid sub-curve $\xi$ as follows. For each obstacle vertex $u_j$ with $j\geq 5$, define $\xi^+_j=\{(s,t) \ |\ f_{u_j}(s,t)> f_{u_1}(s,t)\}$ (we assume $f_{u_j}(s,t)=\infty$ if $s$ is not visible to $u_j$). Then, the common intersection $\bigcup_{j\geq 5}\xi^+_j$ is the portion of $\xi$ over which $\calL(F)$ is attained by the elementary functions $d_{u_i,v_i}(s,t)$, $1\leq i\leq 4$. Doing this for all cells $\sigma$ of $\calO_V$ will determine the edges of $\calL(F)$ defined by the functions $f_{u_i}(s,t)$, $1\leq i\leq 4$. Doing this for all subsets $V$ of four obstacle vertices will compute all interior edges of $\calL(F)$. We next give the details of the algorithm. We first present a preliminary algorithm of $O(n^8\log n)$ time and then improve it to $O(n^{7.73})$ time. 

\paragraph{The first procedure: Finding the valid sub-curves.}
Since $\sigma$ is a segment on $\calB$, $\gamma_{\sigma}$ has $O(1)$ sub-curves whose projections on the $t$-axis are in $\sigma$. Let $\Gamma^*_{\sigma}$ denote the set of these sub-curves, which can be computed in $O(1)$ time. 
%we call these sub-curves {\em candidate-valid sub-curves}. 
%This can be done in the following way. We compute a vertical decomposition of $\sigma$ (i.e., by introducing vertical segments from vertices as well as from tangents of curves) and we also compute a vertical decomposition for $\bbR^2\setminus \sigma$, i.e., the region of the plane outside $\sigma$; we use $\triangle(\sigma)$ to denote the resulting decomposition of the plane. Since $\sigma$ is of size $O(n)$, $\triangle(\sigma)$ can be easily computed in $O(n\log n)$ time, e.g., by a plane sweeping algorithm. Then, by traversing on $\eta_t(\gamma_{\sigma})$ using $\triangle(\sigma)$, we can compute $\Gamma^*_{\sigma}$ in $O(n)$ time (indeed, since $\gamma_{\sigma}$ is of constant size, it intersects each edge of $\triangle(\sigma)$ a constant number of times). Note that $\Gamma^*_{\sigma}=O(n)$. 

Next, for each $1\leq i\leq 3$, we find the sub-curves of $\gamma_{\sigma}$ whose projections on the $s$-plane are inside $\vis(u_i)$; let $\Gamma^i_{\sigma}$ denote the set of these sub-curves. Computing $\Gamma^i_{\sigma}$ can be done in $O(n\log n)$ time, in the same way as in Section~\ref{sec:psipp}. 
%e.g., by first triangulating $\vis(u_i)$ as well as $\bbR^2\setminus \vis(u_i)$ and then traversing on $\eta_s(\gamma_{\sigma})$ using the triangulation. As $\eta_s(\gamma_{\sigma})$ is of constant size, it intersects each edge of $\vis(u_i)$ at most $O(1)$ times. 
%Hence, we can obtain the set $\Gamma^i_{\sigma}$ of $O(n)$ disjoint sub-curves of $\gamma_{\sigma}$. 
Note that $|\Gamma^i_{\sigma}|=O(n)$.

Now we compute the common intersection $\Gamma_{\sigma}$ of the sub-curves of $\Gamma^*_{\sigma}$ and $\Gamma^i_{\sigma}$ for all $1\leq i\leq 3$, i.e., a point $p$ is in a sub-curve of $\Gamma_{\sigma}$ if $p$ is in a curve of $\Gamma^*_{\sigma}$ and also in a curve of $\Gamma^i_{\sigma}$ for every $1\leq i\leq 3$. Since the total number of sub-curves of $\Gamma^*_{\sigma}$ and $\Gamma^i_{\sigma}$, $1\leq i\leq 3$, is $O(n)$, $\Gamma_{\sigma}$ has $O(n)$ sub-curves of $\gamma_{\sigma}$. By definition, $\Gamma_{\sigma}$ is exactly the set of valid sub-curves of $\gamma_{\sigma}$. 
We can compute $\Gamma_{\sigma}$ in $O(n\log n)$ time in the same way as in Section~\ref{sec:psipp}. 

Doing this for all $O(n)$ segments $\sigma$ of the overlay $\calO_V$ computes in $O(n^2\log n)$ time a total of $O(n^2)$ valid sub-curves for $O(n)$ curves $\gamma_{\sigma}$. Let $\Gamma(V)$ denote the set of all of these $O(n^2)$ valid curves. 

This finishes the first procedure, which takes $O(n^2\log n)$ time in total. 

\paragraph{The second procedure: Computing the vertices on valid curves.}
We now compute the vertices of $\calA(F)$ on all valid curves of $\Gamma(V)$. Our approach is based on the fact that a vertex on a curve $\gamma_{\sigma}$ must be the intersection between $\gamma_{\sigma}$ and another elementary function $d_{u_j,v_j}(s,t)$ for some $j>3$. 

For each $u_j$, $j>3$, we do the following. For notational convenience, let $j=4$. 
We compute the overlay $\calO_j$ of $\spmb(u_i)$'s, $1\leq i\leq 4$. For each segment $\sigma'$ of $\calO_j$, for each $1\leq i\leq 4$, $\sigma'$ is contained in a single cell of $\spm(u_i)$ and let $v_i$ denote the root of the cell. 
%(note that $\spm(u_i)$ can be computed in the preprocessing, which takes $O(n^2\log n)$ in total for all obstacle vertices $u_i$). 
Also, $\sigma'$ must be contained in a single segment of $\calO_V$, denoted by $\sigma$. 
We solve the equation $d_{u_1,v_1}(s,t)=d_{u_2,v_2}(s,t)=d_{u_3,v_3}(s,t)=d_{u_4,v_4}(s,t)$ for $s\in \bbR^2$ and $t\in \sigma'$ to obtain $O(1)$ solutions $(s,t)$. For each such solution $(s,t)$, we process it as follows. 
%First, we do a point location to find the cell $\calO_j$ that contains $t$; if the cell is not $\sigma'$, then we ignore this solution. Second, 
We check whether $s\in \vis(u_i)$, for every $1\leq i\leq 4$. If $s$ is not in $\vis(u_i)$ for any $1\leq i\leq 4$, then we ignore this solution $(s,t)$. Otherwise, the point $p=(s,t,d_{u_1,v_1}(s,t))$ must be on the curve $\gamma_{\sigma}$. We further check whether $p$ is on a valid curve of $\Gamma_{\sigma}$ (again in $O(\log n)$ time by binary search). If yes, then $p$ is a vertex and we insert it to the valid curve of $\Gamma_{\sigma}$ that contains it. 

As $\calO_j$ has $O(n)$ segments, the above computes in $O(n\log n)$ time $O(n)$ vertices for the curves of $\Gamma(V)$. Doing this for all $u_j$ with $j>3$ computing all $O(n^2)$ vertices for the curves of $\Gamma(V)$. 
For each curve of $\Gamma(V)$, by sorting all its vertices, we obtain the edges of $\calA(F)$ contained in the curve. As $|\Gamma(V)|=O(n^2)$, this step takes $O(n^2\log n)$ time. Let $E(V)$ denote the set of edges thus computed. 

This finishes the second procedure, which takes $O(n^2\log n)$ time in total. 

\paragraph{The third procedure: Identifying the edges of $\boldsymbol{\calL(F)}$.}
To find the edges of $E(V)$ that are on $\calL(F)$, for each edge $e$ of $E(V)$, pick an interior point $p$ on $e$. Let $s=\eta_s(p)$, $t=\eta_t(p)$, and $d(p)$ be the $4$-th coordinate value of $p$. 
%As argued in Section~\ref{sec:psipp}, 
Observe that $e\in \calL(F)$ if and only if $d(p)=d(s,t)$. 
We compute $d(s,t)$ by a special two-shortest shortest path query in which one of the query point is on $\calB$. For such case, we will show in Corollary~\ref{coro:querybp} in Section~\ref{sec:query} that there is  a data structure of $O(n^5\log n)$ preprocessing time and $O(\log^2 n)$ query time. 
Using that data structure, finding all edges of $E(V)$ that are on $\calL(F)$ can be done in $O(n^2\log^2 n)$ time (excluding the preprocessing time of the data structure). 

Doing this for all subsets $V$ of three obstacle vertices can compute all edges of $\calL(F)$ in $O(n^5\log^2 n)$ time. The following theorem summarizes our result. 

\begin{theorem}\label{theo:algopsipb}
One can compute all edges and vertices of $\calL(F)$ in $O(n^5\log^2 n)$ time. After that, the decomposition $\Psipb$ of size $O(n^5+|\Psipb|)$ can be computed in additional $O((n^5+|\Psipb|)\log n)$ time. 
\end{theorem}

Using a topology curve approach as in \cite{ref:ChiangTw99}, we prove a bound for $|\Psipb|$ in the following. 

\begin{theorem}\label{theo:boundpsipb}
    $|\Psipb|=O(n^8)$. 
\end{theorem}
\begin{proof}
We follow the approach of \cite{ref:ChiangTw99} using topology curves. Note that the complexity of $\Psipb$ was not explicitly studied in \cite{ref:ChiangTw99}. One may consider this proof a straightforward follow-up of the work in \cite{ref:ChiangTw99}.

Suppose $s$ moves inside $\calP$. As $s$ crosses a topology curve, a combinatorial change will happen to $\Psipb$ when there is a point $t\in \calB$ that has three shortest paths from $s$. In the following, we show that there are $O(n^4)$ topology curves each of which is a constant degree algebraic curves of constant size. 

For every three obstacle vertices $v_i$, $1\leq i\leq 3$, we do the following. Recall that $\spmb(v_i)$ has $O(n)$ segments on $\calB$. We overlay $\spmb(v_i)$, for all $1\leq i\leq 3$, resulting in $O(n)$ segments on $\calB$. For each such segment $e$, for each $1\leq i\leq 3$, $e$ is contained in a single cell of $\spm(v_i)$ and let $u_i$ be the root of the cell. Consider the equation $d_{u_1,v_1}(s,t)=d_{u_2,v_2}(s,t)=d_{u_3,v_3}(s,t)$, for $s\in e$ and $t\in\bbR^2$. This equation and $s\in e$ provide three independent constraints in the four variables that are the coordinates of $s$ and $t$, resulting in one degree of freedom; this yields a constant size algebraic curve of constant degree, which we call a {\em topology curve}~\cite{ref:ChiangTw99}. 
%Observe that each second type event with $s\in e$ corresponds to a pair $(s,t)$ with $t$ in the interior of $\calP$ such that there are four shortest \st\ paths (see Fig.~\ref{fig:bpevent10}(b)). If the anchors of $t$ in the four shortest paths are $v_i$, $1\leq i\leq 4$, respectively, then we say that $s$ is {\em defined by} $v_i$, $1\leq i\leq 4$, and $(s,t)$ must be a solution to the above equation. Since the equation has $O(1)$ solutions, $e$ contains $O(1)$ second type events for $s$ defined by $v_i$, $1\leq i\leq 4$. 
As the above overlay has $O(n)$ segments, there are $O(n)$ topology curves defined by $v_i$, $1\leq i\leq 3$. Enumerating all combinations of three obstacle vertices gives the $O(n^4)$ upper bound on the total number of topology curves.     

Since each topology curve is of constant size and constant degree, every two curves have $O(1)$ intersections. Hence, there are $O(n^8)$ intersections between all topology curves, which leads to $|\Psipb|=O(n^8)$.  
\end{proof}

\section{Two-point shortest path queries}
\label{sec:query}
In this section, we study the two-point shortest path query problem. We refer to it as the {\em $\calB$-$\calB$ case} if both query points $s$ and $t$ are required to be on $\calB$, the {\em $\calB$-$\calP$ case} if only one of the query point is required to be on $\calB$, and the {\em general case} otherwise. 

Note that if $s$ is visible to $t$, then $\pi(s,t)$ is $\overline{st}$. It is possible to construct a data structure in $O(n^2)$ space and preprocessing time for $\calP$ such that whether $s$ is visible to $t$ can be determined in $O(\log n)$ time by a ray-shooting query~\cite{ref:PocchiolaGr90,ref:ChenVi15}. In the following, we focus on the case in which $s$ is not visible to $t$.
%\footnote{Note that the visibility between two points $s$ and $t$ does not affect the definitions of the SPM-equivalence decompositions. For example, $\Psipp$ is defined with respect to the combinatorial changes of $\spm(s)$, which is not affected by visibility. In fact, $\Psipp$ is a refinement of {\em visibility-equivalence decomposition}, a subdivision of $P$ into cells in which the visibility polygon of all points in the same cell are combinatorially constant. In applications of the SPM-equivalence decompositions, we sometimes need to handle the visibility in a special (often easier) way such as this two-point shortest path query problem. %For example, when answering two-shortest path queries, we can use a ray-shooting data structure~\cite{ref:PocchiolaGr90,ref:ChenVi15} to handle the case where $s$ and $t$ are visible to each other.}

\subsection{The $\calB$-$\calB$ case}
For the $\calB$-$\calB$ case, we start with constructing the SPM-equivalent decomposition $\Psibb$ in $O(n^{4+\epsilon})$ time by Corollary~\ref{coro:psibb}. Then, for each segment $\sigma$ of $\Psibb$, we construct the shortest path map $\spmb(s)$ for a point $s\in\sigma$. By definition, when $s$ moves in $\sigma$, $\spmb(s)$ does not change combinatorially. Specifically, $\spmb(s)$ partitions $\calB$ into a set $Q_s$ of segments; each segment is in a single cell of $\spm(s)$ and we associate the root of the cell with the segment (called the {\em anchor} of the segment; indeed, it is the anchor of $t$ in the shortest \st\ path for any point $t$ on the segment). As $s$ moves in $\sigma$, the segment endpoints of $\spmb(s)$ may be moving on $\calB$ but each endpoint always stays at the same obstacle edge and the anchor of each segment of $\spmb(s)$ does not change. The endpoint of each segment of $\spmb(s)$ can be represented as a function of the position of $s$ on $\sigma$. In addition, for each obstacle vertex $v$, let $u$ be the anchor of $s$ in the shortest path from $s$ to $v$ on $\spm(s)$. Then, when $s$ moves in $\sigma$, $d(s,v)$ can be represented as a function of the position of $s\in \sigma$, and specifically, $d(s,v)=|su|+d(u,v)$. As such, we can represent $\spmb(s)$ by a parameterized binary search tree, so that given a pair of query points $(s,t)$ with $s\in \sigma$ and $t\in \calB$, we can determine the segment of $\spmb(s)$ containing $t$ in $O(\log n)$ time and then using the anchor of the segment to further compute $d(s,t)$. As such, each such query can be answered in $O(\log n)$ time (an actual shortest path can be output in additional time linear in the number of edges of the path; this can be done by explicitly maintaining the shortest path tree $\spt(u)$ for every obstacle vertex $u$, so that if we know $u$ and $v$ are anchors of $s$ and $t$ in $\pi(s,t)$, respectively, then $\pi(u,v)$ can be output using $\spt(u)$). Note that the idea of building a parameterzied shortest path map was already used by Chiang and Mitchell~\cite{ref:ChiangTw99} for the general case. 

If we maintain a parameterized shortest path map as above for each segment $\sigma$ of $\Psibb$, denoted by $\spmb(\sigma)$, then each $\calB$-$\calB$ case two-point shortest path query can be answered in $O(\log n)$ time. Since each $\spmb(\sigma)$ can be computed in $O(n\log n)$ time and $O(n)$ space, the total preprocessing time and space are $O(|\Psibb|\cdot n \log n)$ and $O(|\Psibb|\cdot n)$, respectively. 
%(the $O(\log n)$ factor is absorbed by $O(n^{\epsilon})$).

We can further improve the preprocessing by a linear factor using persistent trees. Indeed, consider an obstacle edge $e$. For any two adjacent segments $\sigma_1$ and $\sigma_2$ of $\Psibb$ on $e$, observe that $\spmb(\sigma_1)$ and $\spmb(\sigma_2)$ differ by only $O(1)$ changes. Hence, suppose we have a parameterzied binary search tree for the map $\spmb(\sigma_1)$; then $\spmb(\sigma_2)$ can be obtained from $\spmb(\sigma_1)$ with $O(1)$ updates (insertions and deletions). Using the (partially) persistent binary search trees, we can obtain $\spmb(\sigma_2)$ from $\spmb(\sigma_1)$ in $O(\log n)$ time and $O(1)$ space while still keeping $\spmb(\sigma_1)$ for queries~\cite{ref:DriscollMa89} (but no updates are needed in future for $\spmb(\sigma_1)$ 
and thus a partially persistent tree is sufficient). As such, for the edge $e$, we first explicitly construct $\spmb(\sigma_1)$ for the first segment $\sigma_1$ of $\Psibb$ on $e$, which takes 
$O(n\log n)$ time and $O(n)$ space, and then we obtain $\spmb(\sigma)$'s for other segments $\sigma$ of $\Psibb$ on $e$ one by one in the incremental way as described above. Processing all obstacle edges as such takes $O((n^2+|\Psibb|)\log n)$ time and $O(n^2+|\Psibb|)$ space in total. Since $|\Psibb|=\Omega(n^2)$ in the worst case, we have the following result. 

%In this way, the preprocessing time and space can be bounded by $O(|\Psibb|\cdot \log n)$ and $O(|\Psibb|)$, respectively. We thus obtain the following. 

\begin{theorem}
Assuming that $\Psibb$ is available, we can preprocess $\calP$ into a data structure of $O(|\Psibb|)$ space in $O(|\Psibb|\log n)$ time so that for any two query points $s,t\in \calB$, their shortest path length can be computed in $O(\log n)$ time and an actual shortest $s$-$t$ path can be output in additional time linear in the number of edges of the path. 
\end{theorem}

With our algorithm of $O(n^{4+\epsilon})$ time to compute $\Psibb$, we obtain the following. 

\begin{corollary}
We can preprocess $\calP$ into a data structure of $O(n^{4+\epsilon})$ space in $O(n^{4+\epsilon})$ time so that for any two query points $s,t\in \calB$, their shortest path length can be computed in $O(\log n)$ time and an actual shortest $s$-$t$ path can be output in additional time linear in the number of edges of the path.     
\end{corollary}

% \paragraph{The two-obstacle case.}
% For the two obstacle case, where $s$ is required to be on the boundary of an obstacle $P_s$ and $t$ is also required to be on the boundary of an obstacle $P_t$ ($P_s=P_t$ is allowed), we can basically follow the similar approach and obtain a bound analogous to the above theorem with $|\Psibb|$ replaced by $\Psi_{\partial P_1}(\partial P_2)$. Using our upper bound for $\Psi_{\partial P_1}(\partial P_2)$, we have the following result. 

% \begin{corollary}
% With respect to two obstacles $P_s$ and $P_t$ of $\calP$, we can preprocess $\calP$ into a data structure of $O(n^{3+\epsilon})$ space in $O(n^{3+\epsilon})$ time so that for any two query points $s\in \partial P_1$,$t\in \partial P_2$, their shortest path length can be computed in $O(\log n)$ and an actual shortest $s$-$t$ path can be output in additional time linear in the number of edges of the path.     
% \end{corollary}

\subsection{The $\boldsymbol{\calB}$-$\boldsymbol{\calP}$ case}
For the $\calB$-$\calP$ case, we start with constructing the SPM-equivalent decomposition $\Psibp$ in $O((n^2+|\Psibp|)\cdot \log n)$ time by Theorem~\ref{theo:algopsibp}. Recall that $\Psibp$ is a decomposition of $\calB$ into segments so that $\spm(s)$ is topologically constant for all points $s$ on the same segment. 

For each segment $\sigma$ of $\Psibp$, following the method of Chiang and Mitchell~\cite{ref:ChiangTw99}, we construct a parameterized shortest path map $\spm(\sigma)$ for points $s\in \sigma$, which can answer each $\calB$-$\calP$ case two-point shortest path query in $O(\log n)$ time for two query points $s\in \sigma$ and $t\in \calP$ (note that we also need to build a parameterized point location data structure on $\spm(\sigma)$ as discussed in~\cite{ref:ChiangTw99}). As such, after building a parameterized shortest path map for each segment of $\Psibp$, each $\calB$-$\calP$ case two-point shortest path query can be answered in $O(\log n)$ time. 
Constructing a parameterized shortest path map can be done in $O(n\log n)$ time and $O(n)$ space~\cite{ref:ChiangTw99}. 
Hence, the total preprocessing time and space are bounded by $O(|\Psibp|\cdot n\log n)$ and $O(|\Psibp|\cdot n)$ space, respectively. The $O(\log n)$ factor can be reduced by observing that there are only $O(1)$ changes between two shortest path maps of adjacent cells (hence, if we have the map for one cell, obtaining the map for an adjacent cell can be easily done in $O(n)$ time). 

As in~\cite{ref:ChiangTw99}, using persistent data structures, the preprocessing can be further improved by roughly a linear factor, but the query time grows to $O(\log^2 n)$. We summarize the  result below.

\begin{theorem}
\begin{itemize}   
\item We can preprocess $\calP$ into a data structure of $O(|\Psibp|\cdot n)$ space in $O(n^2\log n+ |\Psibp|\cdot n)$ time so that for any two query points $s,t$ with one of them on $\calB$, their shortest path length can be computed in $O(\log n)$ time and an actual shortest $s$-$t$ path can be output in additional time linear in the number of edges of the path. 
\item 
Alternatively, we can preprocess $\calP$ into a data structure of $O(|\Psibp|\cdot \log n)$ space in $O((n^2+|\Psibp|)\cdot \log n)$ time so that for any two query points $s,t$ with one of them on $\calB$, their shortest path length can be computed in $O(\log^2 n)$ time and an actual shortest $s$-$t$ path can be output in additional time linear in the number of edges of the path. 
\end{itemize}
\end{theorem}

Using our $O(n^5)$ upper bound for $|\Psibp|$, we obtain the following result. 
\begin{corollary}\label{coro:querybp}
 \begin{itemize}
\item We can preprocess $\calP$ into a data structure of $O(n^6)$ space in $O(n^6)$ time so that for any two query points $s,t$ with one of them on $\calB$, their shortest path length can be computed in $O(\log n)$ time and an actual shortest $s$-$t$ path can be output in additional time linear in the number of edges of the path. 

\item 
Alternatively, we can preprocess $\calP$ into a data structure of $O(n^5 \log n)$ space in $O(n^5\log n)$ time so that for any two query points $s,t$ with one of them on $\calB$, their shortest path length can be computed in $O(\log^2 n)$ time and an actual shortest $s$-$t$ path can be output in additional time linear in the number of edges of the path. 
\end{itemize}   
\end{corollary}

\subsection{The general case}
For the general case, we first build the SPM-decomposition $\Psi$ in $O(n^{7.73}+|\Psi|\log n)$ time by Corollary~\ref{coro:psipp}. Afterwards, following the approach of Chiang and Mitchell~\cite{ref:ChiangTw99}, for each cell of $\Psi$, we construct the parameterized shortest path map. In this way, each general case two-point shortest path query can be answered in $O(\log n)$ time. The total preprocessing time and space is $O((n^7+|\Psi|)\cdot n)$. As in \cite{ref:ChiangTw99}, using the persistent data structure, the preprocessing time (not including the time for constructing $\Psi$) and space can be reduced to $O((n+|\Psi|)\log n)$, but the query time increases to $O(\log^2 n)$. We summarize the result below. 

\begin{theorem}
\begin{itemize}
\item We can preprocess $\calP$ into a data structure of $O((n^7+|\Psi|)\cdot n)$ space in $O((n^7+|\Psi|)\cdot n)$ time so that for any two query points in $\calP$, their shortest path length can be computed in $O(\log n)$ time and an actual shortest $s$-$t$ path can be output in additional time linear in the number of edges of the path. 

\item 
Alternatively, we can preprocess $\calP$ into a data structure of $O((n+|\Psi|)\cdot \log n)$ space in $O(n^{7.73}+|\Psi|\cdot \log n)$ time so that for any two query points in $\calP$, their shortest path length can be computed in $O(\log^2 n)$ time and an actual shortest $s$-$t$ path can be output in additional time linear in the number of edges of the path. 
\end{itemize}
\end{theorem}

\section{Geodesic diameter}
\label{sec:diameter}

In this section, we discuss the geodesic diameter problem. We use the {\em $\calB$-$\calB$ case} to refer to the problem where we are looking for a diametral pair $(s,t)$ with both $s,t\in \calB$, the {\em $\calB$-$\calP$ case} the problem where $s\in \calB$ and $t\in \calP$, and the {\em general} case where both $s,t\in\calP$. 

\paragraph{The $\boldsymbol{\calB}$-$\boldsymbol{\calB}$ case.}
For the $\calB$-$\calB$ case, suppose that $(s,t)$ is a diametral pair with $s,t\in \calB$. If one of $s$ and $t$, say $s$, is a vertex of $\calP$, then since $t$ must be a vertex of $\spm(s)$~\cite{ref:BaeCo15CGTA,ref:BaeTh13}, such a diametral pair can be easily found in $O(n^2\log n)$ time by computing a shortest path map from each vertex of $\calP$. 

In the following, we assume that neither $s$ nor $t$ is a vertex of $\calP$. In this case, 
Bae, Korman, and Okamoto~\cite{ref:BaeTh13} proved that any diametral pair $(s,t)$ has three topologically different shortest \st\ paths. This observation implies that $(s,t)$ must correspond to 
%vertex in the lower envelope $\calL(F)$ of the set $F$ of functions defined for $\Psibb$ in Section~\ref{sec:psibb}, which also 
an event point in the SPM-decomposition $\Psibb$, i.e., $s$ is the event point while $t$ is the point on $\calB$ where $\spmb(s)$ changes. As such, once $\Psibb$ is computed, all $\calB$-$\calB$ case diametral pairs are available. More specifically, if $s$ is an event point $\Psibb$, then $s$ is created due to a combinatorial change at a point $t$ on $\calB$; we keep $(s,t)$ as a {\em candidate} diametral pair (and also keep their geodesic distance $d(s,t)$). We can slightly modify our algorithm for constructing $\Psibb$ to compute all these candidate diametral pairs and their geodesic distances. The total time of the algorithm is still $O(|\Psibb|\log n)$. After having all these candidate diametral pairs, we simply return all pairs whose geodesic distances are the largest. 

%Using our algorithm for constructing $\Psibb$, we immediately have the following result. 

\begin{theorem}
All $\calB$-$\calB$ case diametral pairs can be computed in $O(|\Psibb|\log n)$ time. 
\end{theorem}

Using our bound $|\Psibb|=O(n^{4+\epsilon})$, we have the following. 

\begin{corollary}
All $\calB$-$\calB$ case diametral pairs can be computed in $O(n^{4+\epsilon})$ time.     
\end{corollary}

\paragraph{The $\boldsymbol{\calB}$-$\boldsymbol{\calP}$ case.}
For the $\calB$-$\calP$ case, suppose that $(s,t)$ is a diametral pair with $s\in \calB$. 
If $s$ is a vertex of $\calP$, then since $t$ must be a vertex of $\spm(s)$~\cite{ref:BaeCo15CGTA,ref:BaeTh13}, such a diametral pair can be easily found in $O(n^2\log n)$ time by computing a shortest path map from each vertex of $\calP$. 
If $t\in \calB$, then it becomes the $\calB$-$\calB$ case and thus such a diametral pair can be computed as discussed above. 

We now assume that $s$ is in the interior of an edge of $\calB$ and $t$ is in the interior of $\calP$. In this case, 
Bae, Korman, and Okamoto~\cite{ref:BaeTh13} proved that any diametral pair $(s,t)$ has four topologically different shortest \st\ paths. This implies that $(s,t)$ must 
%correspond to a vertex in the lower envelope $\calL(F)$ of the set $F$ of functions defined for $\Psibp$ in Section~\ref{sec:psibb}, which also 
correspond to an event point in the SPM-decomposition $\Psibp$. As in the $\calB$-$\calB$ case, we can modify the algorithm for $\Psibp$ to compute all $\calB$-$\calP$ case diametral pairs in $O(|\Psibp|\log n)$ time. 
%Using our algorithm for constructing $\Psibb$, we immediately have the following result.

\begin{theorem}
All $\calB$-$\calP$ case diametral pairs can be computed in $O(|\Psibp|\log n)$ time. 
\end{theorem}

With our bound $|\Psibp|=O(n^5)$, we obtain the following result. 

\begin{corollary}
All $\calB$-$\calP$ case diametral pairs can be computed in $O(n^{5}\log n)$ time.     
\end{corollary}

\paragraph{The general case.}
For the general case, Bae, Korman, and Okamoto~\cite{ref:BaeTh13} proved that any diametral pair $(s,t)$ has five topologically different shortest \st\ paths. This observation implies that $(s,t)$ must correspond to a vertex in the lower envelope $\calL(F)$ of the set $F$ of functions defined for $\Psipp$ in Section~\ref{sec:psipp}. As such, once the vertices of $\calL(F)$ are computed, all general case diametral pairs are available. Hence, we have the following result. 

\begin{theorem}
If all vertices of the lower envelope $\calL(F)$ of the set $F$ of functions defined for $\Psipp$ can be computed in $O(T)$ time, then 
all general case diametral pairs can be computed in $O(T)$ time. 
\end{theorem}

Our algorithm for constructing $\calL(F)$ runs in $O(n^{7.73})$ time, which matches the runtime of the algorithm in~\cite{ref:BaeTh13}.

\section{Geodesic center}
\label{sec:center}
In this section, we discuss the geodesic center problem. For two subsets $A,B\subseteq \calP$, we use the {\em $A$-$B$ case} to refer to the problem in which one wishes to find a point $s\in A$ that minimizes the largest geodesic distance from $s$ to any point $t\in B$; we call $s$ an {\em $A$-restricted center with respect to $B$} and the geodesic distance between $s$ and its farthest point in $B$ is the corresponding {\em geodesic radius}. For example, the $\calP$-$\calB$ case refers to the problem of computing
a point $s\in \calP$ that minimizes the largest geodesic distance from $s$ to any point $t\in \calB$. In particular, the {\em general case} refers to the $\calP$-$\calP$ case. 

%We use {\em the $\calB$-$\calB$ case} to refer to the problem where one wants to find a point $s\in \calB$ that minimizes the farthest point on $\calB$ from $s$. We use {\em the $\calB$-$\calP$ case} to refer to the problem where one wants to find a point $s\in \calB$ that minimizes the farthest point on $\calP$ from $s$. We use {\em the $\calP$-$\calB$ case} to refer to the problem where one wants to find a point $s\in \calP$ that minimizes the farthest point on $\calB$ from $s$. We use {\em the general case} to refer to the problem where one wants to find a point $s\in \calP$ that minimizes the farthest point on $\calB$ from $s$.

We discuss the following cases in order: the $\calB$-$\calP$ case, the $\calB$-$\calB$ case, the $\calP$-$\calB$ case, and the general case. 

\paragraph{The $\boldsymbol{\calB}$-$\boldsymbol{\calP}$ case.}
We first consider the $\calB$-$\calP$ case (the algorithm for the $\calB$-$\calB$ case is very similar, but simpler). A useful observation that was proved by Bae, Korman, and Okamoto~\cite{ref:BaeCo15CGTA} is the following: For any point $s\in \calP$, any of its farthest points must be a vertex of $\spm(s)$.

We first construct the SPM-equivalent decomposition $\Psibp$. For each segment $\sigma$ of $\Psibp$, by definition, when $s$ moves on $\sigma$, $\spm(s)$ does not change topologically, and in particular, $\spm(s)$ has the same set of vertices. As $s$ moves on $\sigma$, for each vertex $t\in \spm(s)$, the geodesic distance $d(s,t)$ is a constant-degree univariate function of the position of $s\in \sigma$. Let $\calU(\sigma)$ denote the upper envelope of the functions of all vertices $t$ of $\spm(s)$. We consider the problem of finding a $\sigma$-restricted center $s\in \sigma$ with respect to $\calP$ and its corresponding geodesic radius. Observe that such a center must correspond to a lowest vertex in $\calU(\sigma)$ and its geodesic radius is the height of the vertex. As such, an $\sigma$-restricted center and its radius can be easily found once $\calU(\sigma)$ is constructed. The latter task can be done in $O(\lambda_{O(1)}(n)\log n)$ time~\cite{ref:SharirDa95} since there are $O(n)$ constant-degree algebraic functions of constant size, where $\lambda_{O(1)}(n)$ is the maximum length of a $DS(n,O(1))$-sequence and is bounded by $O(n\log^*n)$ (for notational simplicity, we will use the bound $O(n\log^*n)$ instead in the following). Constructing $\calU(\sigma)$ for all segments $\sigma\in \Psibp$ will compute
all $\sigma$-restricted centers and their corresponding geodesic radii. Among them, we return those whose radii are the smallest.  
The total time is thus $O(|\Psibp|\cdot n\log n\log^*n)$, which is bounded by $O(n^6\log n\log^*n)$ using our bound $|\Psibp|=O(n^5)$. 

We now further improve the runtime by a linear factor. For each segment $\sigma\in \Psibb$, let $F(\sigma)$ denote the set of functions $d(s,t)$ for all vertices $t$ of $\spm(\sigma)$, where $\spm(\sigma)$ denote the shortest path map $\spm(s)$ for points $s\in \sigma$. Our main idea is based on the observation that for two adjacent segments $\sigma_1$ and $\sigma_2$ of $\Psibp$, their shortest path maps differs by at most $O(1)$ vertices. Hence, the set difference between $F(\sigma_1)$ and $F(\sigma_2)$ is $O(1)$. The details are given below. 

Consider an obstacle edge $e$. We consider the problem of finding an $e$-restricted center $s\in e$ with respect to $\calP$ and its corresponding geodesic radius. We consider the segments of $\Psibp$ on $e$ in order. Let $\sigma_1$ and $\sigma_2$ be the first two segments. Suppose we already have the functions $F(\sigma_1)$, defined on $\sigma_1$. For each function in $F(\sigma_1)\cap F(\sigma_2)$, instead of creating a new function for $F(\sigma_2)$, we simply extend the domain of the function from $s\in \sigma_1$ to $s\in \sigma_1\cup \sigma_2$. 
In this way, since $F(\sigma_1)$ and $F(\sigma_2)$ differs by $O(1)$, we only need to create $O(1)$ new functions for $F(\sigma_2)$. As such, the total number of functions of all segments of $\Psibp$ on $e$ is $O(n+\kappa(e))$, where $\kappa(e)$ is the number of segments of $\Psibp$ on $e$. Each $e$-restricted geodesic center corresponds to a lowest vertex of the upper envelope $\calU(e)$ of all these functions (again the height of the vertex is the corresponding geodesic radius). Since these are $O(n+\kappa(e))$  constant-degree functions of constant size, the upper envelope $\calU(e)$ can be computed in $O((n+\kappa(e))\log n\log^*n)$ time. Doing this for all obstacle edges $e$ will compute all $\calB$-$\calP$ case geodesic centers (i.e., among all edge restricted centers, we return those with smallest geodesic radii). Since $\sum_{e\in \calB}\kappa(e)=|\Psibp|$, the total time is bounded by $O((n^2+|\Psibp|)\log n\log^*n)$. 
As $|\Psibp|=\Omega(n^2)$ in the worst case, we obtain the following result. 

\begin{theorem}
All $\calB$-$\calP$ case geodesic centers can be computed in $O(|\Psibp|\log n\log^*n)$ time. 
\end{theorem}

Using the $O(n^5)$ bound for $|\Psibp|$, we have the following. 

\begin{corollary}
All $\calB$-$\calP$ case geodesic centers can be computed in $O(n^5\log n\log^*n)$ time. 
\end{corollary}

\paragraph{The $\boldsymbol{\calB}$-$\boldsymbol{\calB}$ case.}
For the $\calB$-$\calB$ case geodesic centers, we follow the similar approach. The difference is the we use the decomposition $\Psibb$ instead. We omit the details and only summarize the result below. 

\begin{theorem}
All $\calB$-$\calB$ case geodesic centers can be computed in $O(|\Psibb|\log n\log^*n)$ time. 
\end{theorem}

Using the $O(n^{4+\epsilon})$ bound for $|\Psibb|$, we obtain the following result (note that the $n^{\epsilon}$ factor absorbs the logarithmic factor). 

\begin{corollary}
All $\calB$-$\calB$ case geodesic centers can be computed in $O(n^{4+\epsilon})$ time. 
\end{corollary}

\paragraph{The $\boldsymbol{\calP}$-$\boldsymbol{\calB}$ case.}
We now discuss the $\calP$-$\calB$ case, which is to find a geodesic center $s\in \calP$ with respect to $\calB$. We can basically follow the same idea, with a major difference: We have to deal with bivariate functions (instead of the univariate functions in the above two cases). 

Specifically, we first construct the SPM-equivalent decomposition $\Psipb$ in $O((n^5+|\Psipb|)\log n)$ time by Theorem~\ref{theo:algopsipb}. We further compute a vertical decomposition of $\Psipb$ (i.e., by extending from each vertex and each $x$-extreme point on a curve a vertical segment until the cell boundaries) so that each cell becomes constant size. Note that the combinatorial complexity of $\Psipb$ does not change asymptotically. 
Then, for each cell $\sigma$ of $\Psipb$, when $s$ moves in $\sigma$, $d(s,t)$ is a constant-degree bivariate function of $s\in \sigma$, for each vertex of $\spmb(s)$; also, since $\sigma$ is of constant size, $d(s,t)$ is a constant-sized function. We wish to find a $\sigma$-restricted center $s\in \sigma$ with respect to $\calB$. Let $F$ be the set of functions $d(s,t)$ for all vertices $t\in \spmb(s)$; hence, $|F|=O(n)$. A $\sigma$-restricted center corresponds to a lowest vertex on the upper envelope $\calU(F)$ of all functions of $F$ and the height of the vertex is the corresponding geodesic radius. Since each function of $F$ is bivariate, constant-degree, and constant-sized, constructing $\calU(F)$ can be done in $O(n^{2+\epsilon})$ time~\cite{ref:AgarwalTh96}. As such, all $\sigma$-restricted geodesic centers can be found in $O(n^{2+\epsilon})$ time. Doing this for all cells of $\Psipb$ will find all $\calP$-$\calB$ case geodesic centers in $O((n^5+|\Psipb|)\cdot n^{2+\epsilon})$ time. We thus have the following result. 

\begin{theorem}\label{theo:diameterpb}
All $\calP$-$\calB$ case geodesic centers can be computed in $O((n^5 +|\Psipb|)\cdot n^{2+\epsilon})$ time. 
\end{theorem}

Using the $O(n^8)$ bound for $|\Psipb|$ by Theorem~\ref{theo:boundpsipb}, we obtain the following result. 

\begin{corollary}
All $\calP$-$\calB$ case geodesic centers can be computed in $O(n^{10+\epsilon})$ time. 
\end{corollary}

\paragraph{Remark.}
One may wonder whether it is possible to improve the algorithm by making use of the property that shortest path maps between two adjacent cells of $\Psipb$ have only $O(1)$ differences. However, we find difficulty in doing so. Indeed, suppose $\sigma_1$ and $\sigma_2$ are two adjacent cells of $\Psipb$. Let $F(\sigma_1)$ and $F(\sigma_2)$ be the corresponding sets of functions. If a function is in both $F(\sigma_1)$ and $F(\sigma_2)$, we wish to extend the function from $F(\sigma_1)$ to $F(\sigma_2)$. The issue is that this extension increases the combinatorial size of the domain of $s$ from $\sigma_1$ to $\sigma_1\cup \sigma_2$. Since each cell is two-dimensional, changing the domain size will not guarantee the constant size of the function any more (but this issues does not exist in the 1D case, i.e., in both the $\calB$-$\calB$ case and the $\calB$-$\calP$ case; indeed, in these two cases, $\sigma_1$ and $\sigma_2$ are two segments with a common endpoint on an obstacle edge $e$, and thus $\sigma_1\cup \sigma_2$ is still a single segment of $e$).

\paragraph{The general case.}
The general case follows the same idea as the above $\calP$-$\calB$ case, but uses the decomposition $\Psi$ instead. 
%Since constructing $\Psi$ takes $O((n^8+|\Psi|)\log n)$ time, 
With Corollary~\ref{coro:psipp}, we  have the following result. 

\begin{theorem}\label{theo:diameterpp}
All general case geodesic centers can be computed in $O((n^7+|\Psi|)\cdot n^{2+\epsilon})$ time. 
\end{theorem}

The currently best upper bound for $|\Psi|$ is $O(n^{10})$~\cite{ref:ChiangTw99}, which makes the runtime of the algorithm in Theorem~\ref{theo:diameterpp} bounded by $O(n^{12+\epsilon})$. This matches the result obtained by Bae, Korman, and Okamoto~\cite{ref:BaeCo15CGTA}, but is worse than the $O(n^{11}\log n)$ time solution by Wang~\cite{ref:WangOn18}. 
%Our contribution is to reduce the problem of designing efficient algorithms to proving smaller bounds for $|\Psi|$. 

%%% >>> references
% \footnotesize
%\bibliographystyle{abbrv}
\bibliographystyle{plainurl}
\bibliography{reference}
%%% <<< end of reference

% \appendix
% \section*{APPENDIX} 
% \label{sec:appendix}

%\section{Proof of Lemma~\ref{lem:10}}
%\label{app:proof-lemma-10}

%\begin{figure}[t]
%\begin{minipage}[t]{0.5\textwidth}
%\begin{center}
%\includegraphics[height=1.3in]{annuli.pdf}
%\caption{\footnotesize An annulus $D_p$ (the grey region).}
%\label{fig:annuli}
%\end{center}

%\end{minipage}
%\hspace{0.05in}
%\begin{minipage}[t]{0.48\textwidth}
%\begin{center}
%\includegraphics[height=1.2in]{pseudotrap.pdf}
%\caption{\footnotesize Illustrating a pseudo-trapezoid.}
%\label{fig:pseudotrap}
%\end{center}
%\end{minipage}
%\vspace{-0.1in}
%\end{figure}

\end{document}